\title{Secret Sharing Schemes from Correlated Random Variables and Rate-Limited Public Communication}
 \author{Rumia Sultana, R\'{e}mi A. Chou\\ \thanks{R. Sultana is with the EECS Department at Embry-riddle Aeronautical University. R. Chou is with the CSE Department at The University of Texas at Arlington.  This work was supported in part by NSF grants CCF-1850227 and CCF-2401373. Part of the results has been presented at the 2021 IEEE International Symposium on Information Theory~\cite{sultana2021low}.
 E-mails: sultanr1@erau.edu,  remi.chou@uta.edu. }
}
\date{Jan 2020}
\documentclass[10pt,article]{IEEEtran}
\IEEEoverridecommandlockouts
\usepackage{epsfig,latexsym,amsmath,amssymb,setspace,cite,color,graphicx,algorithm,algorithmic,cases}
\usepackage[caption=false,font=footnotesize]{subfig}
\usepackage[percent]{overpic}

\DeclareFontFamily{OT1}{pzc}{}
\DeclareFontShape{OT1}{pzc}{m}{it}{<-> s * [1.10] pzcmi7t}{}
\DeclareMathAlphabet{\mathpzc}{OT1}{pzc}{m}{it}
\usepackage{amsmath}
\usepackage{amsfonts}
\usepackage{bookmark}
\usepackage{amssymb}
\usepackage{dsfont}
\usepackage{bm}
\usepackage{amsthm}
\usepackage{newlfont}
\usepackage{float}
\usepackage{hyperref}
\usepackage{algorithm}
\usepackage{algorithmic}
\usepackage{enumerate}
\usepackage{chngcntr}
\usepackage{breqn}
\usepackage{stmaryrd}
\usepackage{cite}
\usepackage{amsmath}
\usepackage{amsfonts}
\usepackage{amssymb}
\usepackage{dsfont}
\usepackage{bbm}
\usepackage{amsthm}
\usepackage{newlfont}
\usepackage{float}
\usepackage{hyperref}
\usepackage{enumerate}
\usepackage{chngcntr}
\usepackage{mathtools}
\usepackage{breqn}
\usepackage{stmaryrd}
\usepackage{cite}
\usepackage{algorithm}
\usepackage{algorithmic}
\newtheorem{definition}{Definition}
\newtheorem{corollary}{Corollary}
\newtheorem{lem}{Lemma}
\newtheorem{thm}{Theorem}
\newtheorem{prop}{Proposition}
\newtheorem{ex}{Example}

\allowdisplaybreaks
\hypersetup{
    colorlinks=true, 
    linkcolor= blue,
    citecolor=blue 
}
\usepackage{todonotes}

\usepackage{amsmath} \DeclareMathOperator*{\argmax}{arg\,\!max}

\begin{document}
\maketitle
\begin{abstract}
A dealer aims to share a secret with participants so that only predefined subsets can reconstruct it, while others learn nothing. The dealer and participants access correlated randomness and communicate over a one-way, public, rate-limited channel. For this problem, we propose the first explicit coding scheme able to handle arbitrary access structures and achieve the best known achievable rates, previously obtained non-constructively. Our construction relies on lossy source coding coupled with distribution approximation to handle the reliability constraints, followed by universal hashing to handle the security constraints. 
As a by-product, our construction also yields explicit coding schemes for secret-key generation under one-way, rate-limited public communication that, unlike prior work, achieves the capacity for arbitrary source correlations and do not require a pre-shared secret to ensure strong secrecy.
\end{abstract}

\section{Introduction}
Secret sharing was first introduced in \cite{shamir1979share} and \cite{blakley1979safeguarding}, where a dealer distributes a secret among participants such that only predefined subsets can recover it, while any other sets of colluding participants learn nothing about the secret. A classical example is a multi-person authorization system, where several authorized individuals must jointly approve a sensitive operation, such as launching a missile, so that no single individual can execute the action alone. In this paper, unlike in~\cite{shamir1979share,blakley1979safeguarding}, we consider a secret sharing problem where noisy resources are available to the dealer and the participants. Specifically, the participants and dealer have access to samples of correlated random variables (e.g., obtained in a wireless communication network from channel gain measurements  after  appropriate  manipulations~\cite{wilson2007channel,wallace2010automatic,ye2010information,pierrot2013experimental}), in addition to a one-way (from the dealer to the participants),  public, and rate-limited communication channel. 
Such secret sharing models that rely on noisy resources have been introduced in \cite{zou2015information} for wireless channels and in \cite{csiszar2010capacity,chou2018secret} for source models to avoid the  assumption, made in \cite{shamir1979share,blakley1979safeguarding}, that individual secure channels are available for free between each participant and the dealer. An example application is two-factor authentication, where the two factors correspond to noisy physical sources, for instance, two wireless or sensor-based measurements. Authentication succeeds only when both factors are jointly verified, while each alone provides no information about the secret. This two-factor verification  ensures that compromising a single source is insufficient to gain access.

\subsection{Contributions}

In this paper, we propose the first explicit secret sharing scheme for source models with arbitrary access structures and rate-limited public communication. 

Our construction relies on two coding layers. The first  layer handles the reliability constraints via vector quantization implemented with polar codes \cite{arikan2010source} to support rate-limited public communication. The second  layer handles the security constraints via universal hashing \cite{carter1979universal}. While two-layer constructions have been used in previous work (reviewed in the next section) for simpler settings, a main difficulty in this work is the analysis of the security guarantees after combining the first and second coding layers. Specifically,  one needs to carefully design the first layer to ensure it possesses the appropriate properties for enabling a tractable analysis when combined with the second layer.

While the first coding layer is related to Wyner-Ziv coding, for which explicit polar coding schemes exist, e.g., \cite{korada2010polar}, in our model we face two  challenges: (i) one needs to precisely control the distribution output of the encoder to enable a precise analysis of the security guarantees when one combines the first coding layer with the second coding layer, (ii) unlike the standard problem of Wyner-Ziv coding,  because of the presence of an access structure in our setting, the statistics of the side information available at the decoder is not fully known at the encoder but only known to belong to a given set of probability distributions. One can partially reuse a Block-Markov coding idea from \cite{ye2014universal}, which considers the simpler problem of lossless source coding with compound side information, but we need to develop a novel scheme able to simultaneously perform lossy source coding with compound side information and precisely control the encoder distribution~output.

As a by-product, our construction yields explicit and capacity-achieving coding schemes for the related problem of secret-key generation from correlated randomness and rate-limited one-way public communication \cite{csiszar2000common}. However, unlike prior explicit coding schemes 
 \cite{chou2013polar},  our construction does \emph{not} need a pre-shared secret to ensure strong~secrecy, and achieves the capacity for arbitrary source correlations. 

\subsection{Related works}
Secret sharing models with noisy  resources  available  to  the  dealer  and  the  participants have been studied for
channel models \cite{zou2015information} and source models  \cite{csiszar2010capacity,chou2018secret,chou2019biometric,rana2020secret,chou2021distributed,miller2026secret}, which are related to compound wiretap channels \cite{liang2009compound} and compound secret-key generation  \cite{tavangaran2016secret,bloch2010channel, chou2013secret,yachongka2024secret,yachongka2025secret}, respectively, in
that multiple reliability and security constraints need to be
satisfied simultaneously. However,  all these references only prove the existence of coding schemes, whereas in this paper we focus on the construction of explicit coding schemes. 

While no explicit coding scheme has been proposed in the literature for  secret sharing source models with arbitrary access structures, several works have focused on the simpler problem of secret-key generation between two parties from correlated random variables and public communication~\cite{maurer1993secret},~\cite{ahlswede1993common}. Specifically, explicit coding schemes that achieve optimal secret-key rates for this problem have been developed in the case of rate-unlimited public communication by successively handling the reliability requirement and the secrecy requirement  { employing} source coding with side information and universal hashing, respectively \cite{bennett1995generalized,cachin1997linking,maurer2000information}. However, extending their results to the rate-limited public communication setting \cite{chou2014separation,nitinawarat2012secret} requires vector quantization, for which explicit constructions are challenging and have not been proposed in this context. Consequently, these prior works cannot address the rate-limited communication scenario considered in this paper.

Still considering secret-key generation between two parties, another approach  jointly handles the reliability and secrecy requirements via polar codes and yields optimal key rates for rate-unlimited communication~\cite{chou2013polar,renes2013efficient} and rate-limited communication~\cite{chou2013polar}. However, these works do not seem to easily extend to our setting as   \cite{chou2013polar,renes2013efficient} only consider two parties, and in the case of rate-limited public communication,~\cite{chou2013polar} requires a pre-shared secret between the users. While this pre-shared secret has a negligible rate in~\cite{chou2013polar}, such a resource is forbidden in this~paper.
  
 Note that explicit coding schemes for secret-key generation involving more than two parties have also been proposed in~\cite{nitinawarat2010secret,nitinawarat2010perfect,chou2013polar,chou2019secret} but only when the correlations of the random variables observed by the participants have specific structures. Hence, these works cannot be applied to our setting as we consider a source with arbitrary correlations. 
  
  Finally, note that \cite{chou2020unified}  provides explicit coding schemes for secret sharing channel models with arbitrary access structures~\cite{chou2020unified}. But, again, the coding scheme in \cite{chou2020unified} cannot be applied here as it cannot handle rate-limited communication.

\subsection{Organization of the paper}
Section~\ref{sec:problem statement} presents the problem formulation, and Section~\ref{sec:main results} states the main results. An auxiliary result required for the scheme construction and its proof is developed in Sections~\ref{sec:auxiliary result} and~\ref{sec:proof of theorem 3}, respectively. The proofs of the main results are provided in Sections~\ref{sec:proof of theorem 1} and~\ref{proof of theorem 2}. Concluding remarks are given in Section~\ref{concluding remarks}.

\section{Notation}
For $a,b\in \mathbb{R}$, define $\llbracket a,b \rrbracket \triangleq [\lfloor a \rfloor, \lceil b \rceil] \cap \mathbb{N}$, $[ a] \triangleq \llbracket  1,a \rrbracket$, and  $[a]^+ \triangleq \max (0,a)$. 
The components of a vector $X^{1:N}$ of length $N \in \mathbb{N}$ are denoted with superscripts, i.e., $X^{1:N} \triangleq (X^1,X^2,\ldots,X^{N})$. For any set $\mathcal{A} \subset [N]$, let $X^{1:N}[\mathcal{A}]$   be the components of $X^{1:N}$ whose indices are in $\mathcal{A}$. For two  probability  distributions $p_X$ and $q_X$ defined  over  the  same alphabet
$\mathcal X$, define the  relative entropy 
$
   \mathbb{D}(p_X \lVert q_X)\triangleq  \sum_{x \in \mathcal{X}} p_X(x) \log \frac{p_X(x)}{q_X(x)},
$ and the variational distance 
$    \mathbb{V}(p_X,q_X)\triangleq \sum_{x \in \mathcal{X}}| p_X(x)-q_X(x)|.
$
For two  probability distributions $p_{XY}$ and $q_{XY}$, both  defined  over 
$\mathcal X \times \mathcal{Y}$, define the  conditional relative entropy between $p_{Y|X}$ and $q_{Y|X}$ as 
$
  \mathbb {D}( p_{Y|X} \lVert q_{Y|X})\triangleq  \sum_{x\in \mathcal{X}} p_X(x)  \sum_{y \in \mathcal{Y} } p_{Y|X}(y|x)\log \frac{p_{Y|X}(y|x)}{q_{Y|X}(y|x)} .
$
Let $2^{\mathcal{S}}$ denote the power set of $\mathcal{S}$. Let $\bigtimes$ denote the Cartesian product. 
\label{sec:notation}
\section{Problem Statement} 

 Consider a dealer, $J$ participants, and a discrete memoryless source  $(\mathcal X\times \mathcal Y_{[J]},p_{XY_{[J]}}) $, where $\mathcal{X} \triangleq \{ 0,1\}$,  $\mathcal{Y}_{[J]} \triangleq \bigtimes_{j\in [J]} \mathcal{Y}_j$ with $(\mathcal{Y}_j)_{j \in [J]}$  $J$ finite alphabets.  Let $\mathbb{A}$ be a set of subsets of $[J]$ such that if $\mathcal{T}\subseteq [J]$ contains a set that belongs to $\mathbb{A}$, then $\mathcal{T}$ also belongs to $\mathbb{A}$, i.e., $\mathbb{A}$ is a monotone access structure~\cite{benaloh1988generalized}. Then, define $\mathbb{U}\subseteq 2^{[J]}\backslash\mathbb{A}$ as the set of all colluding subsets of users who must not learn any information about the secret. For any $\mathcal{U}\in \mathbb{U}$ and $\mathcal{A}\in \mathbb{A}$, we use the notation $Y^{1:N}_{\mathcal{U}}\triangleq (Y^{1:N}_j)_{j\in \mathcal{U}}$ and $Y^{1:N}_{\mathcal{A}}\triangleq (Y^{1:N}_j)_{j\in \mathcal{A}}$. Moreover, we assume that
the dealer can communicate with the participants over a one-way, rate-limited,
noiseless,\footnote{If the public noiseless channel is replaced by a noisy one, channel coding can be used to recover a noiseless link.}  and public communication channel.

             \begin{definition}
              A $(2^{NR_{s}}, R_p, \mathbb{A},\mathbb{U}, N)$ secret sharing scheme  consists of:
            \begin{itemize}
                \item A public message alphabet $\mathcal{M} \triangleq [2^{NR_p}]$;
                \item An alphabet $\mathcal{S}\triangleq [2^{NR_s}]$;
                \item Two encoding functions $f : \mathcal X^{N} \rightarrow  \mathcal{M}$ and $h : \mathcal X^{N} \rightarrow  \mathcal{S}$;
                \item $|\mathbb{A}|$ 
                      decoding functions $g_{\mathcal{A}} : \mathcal M \times \mathcal Y_{\mathcal{A}}^{N} \rightarrow \mathcal S  $ for $\mathcal{A}\in \mathbb{A}$;
            \end{itemize}
            and operates as follows
\begin{enumerate}
 \item The dealer observes 
 $ X^{1:N}$ and  Participant $j \in [J]$ observes $Y_j^{1:N}$;             
\item The dealer transmits $M \triangleq f(X^{1:N}) \in \mathcal{M}$ over the public communication channel; 
\item The dealer computes the secret $S\triangleq h(X^{1:N})\in \mathcal{S}$;
\item Any subset of participants $\mathcal{A}\in \mathbb{A}$
                      can compute an estimate of $S$ as $\hat S(\mathcal{A}) \triangleq g_{\mathcal{A}}(M,Y_{\mathcal{A} }^{1:N})$. 
                  \end{enumerate}\label{dfn3}
              \end{definition}
                \begin{definition}
                A rate  $R_s$ is achievable if there exists a sequence of $(2^{NR_s}, R_p, \mathbb{A},\mathbb{U}, N)$ secret sharing schemes  
                such~that 
                     \begin{align}
           \displaystyle \lim_{N \rightarrow +\infty} \displaystyle\max_{\mathcal{A} \in \mathbb{A}}\mathbb{P}[\hat S(\mathcal{A})\neq S ] &=0,\phantom{-}(\text{Reliability})\label{eqnrel}\\
                   \displaystyle \lim_{N \rightarrow +\infty}  \displaystyle\max_{\mathcal{U}\in \mathbb{U}} 
             I(S;MY^{{1:N}}_{\mathcal{U}}  
             ) &=0,\phantom{-}(\text{Strong Security})\label{eqnsec}\\
                  \displaystyle \lim_{N \rightarrow +\infty}  \log |\mathcal{S} 
                 |-H(S) &=0.\phantom{-}(\text{Secret Uniformity})\label{eqn23}
         \end{align} $(\ref{eqnrel})$ means that any subset of participants in $\mathbb{A}$ is able to recover the secret, $(\ref{eqnsec})$ means that any subset of participants in $\mathbb{U}$ cannot obtain information about the secret, while $(\ref{eqn23})$ means that
the secret is nearly uniform. The secret capacity $C_s( R_p)$  is defined as the supremum of all achievable rates.
               \end{definition}\label{sec:problem statement}
               The setting is illustrated in Figure~\ref{figFGD} for $J=3$ participants, where the qualified sets are $\mathbb{A } = \{\{1, 2\}, \{2, 3\}, \{1, 2, 3\}\}$ and the unqualified sets are $\mathbb{U}= 2^{[J]}\backslash \mathbb{A }$. In this example, the dealer wishes to share a single secret $S$ among three participants such that only the subsets in $\mathbb{A}$ can recover it, while any subset in $\mathbb{U}$ learns no information about it.
Specifically, Participants $1$ and $2$, or Participants $2$ and $3$, can collaborate by pooling their shares to reconstruct the secret $S$. In contrast, Participant~$1$ or Participant $3$ alone, remains unqualified and gains no information about the secret. The set of participants $\{1,2,3\}$ is also qualified, since if a smaller qualified subset can recover the secret, so can any larger subset that contains it.
               \begin{figure}
\centering
 \includegraphics[width=8.9 cm]{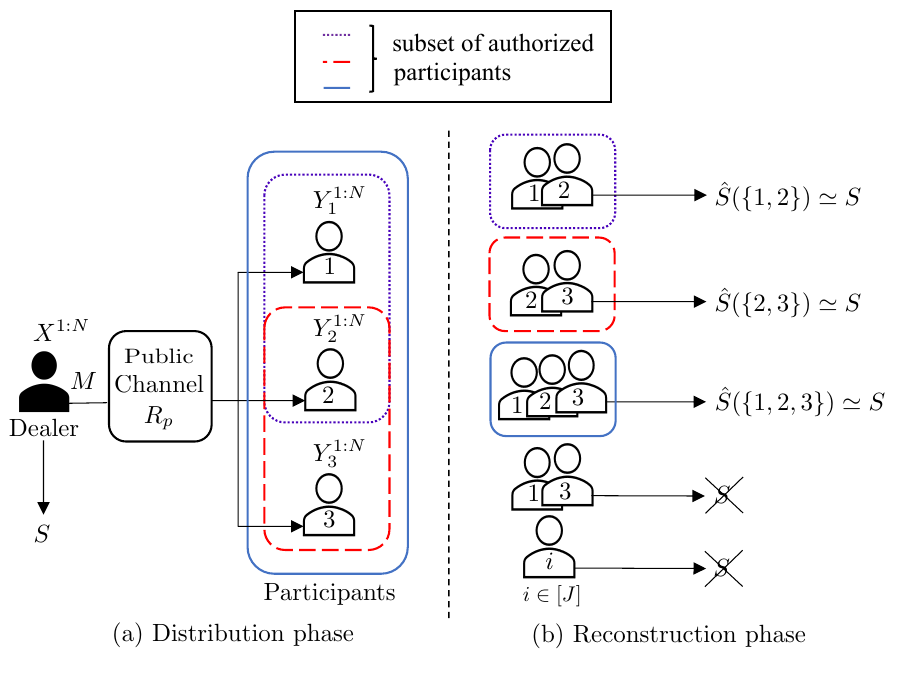}
 \caption{Secret sharing from correlated randomness and rate-limited public communication with $J=3$, $\mathbb{A } = \{\{1, 2\}, \{2, 3\}, \{1, 2, 3\}\}$, $\mathbb{U}= 2^{[J]}\backslash \mathbb{A }  = \{\{1, 3\}, \{1\}, \{2\}, \{3\}\}$. }
    \label{figFGD}
\end{figure}
               \section{Main Results}
               Section \ref{sec:main results} gives our main results for the secret sharing setting of Section \ref{sec:problem statement}. Section \ref{sec:key} applies our result to the problem of secret-key generation \cite{ahlswede1993common,maurer1993secret,csiszar2000common}.
          \subsection{Secret Sharing}
 \begin{prop}The coding scheme of Section \ref{seccodingscheme} achieves the secret  rate
\begin{align*}\nonumber
  R_s=&\displaystyle{\max_{\substack{U\\   U-X-Y_{[J]}}}}\left[\displaystyle{\min_{\mathcal{A}\in \mathbb{A}}}I({U;{Y_{\mathcal{A}}}})- {\displaystyle \max_{\mathcal{U}\in \mathbb{U}}}I(U;Y_{\mathcal{U}})\right]^+\\
&\phantom{-lll-} \text{subject to}\phantom{-} R_p \geq
   \displaystyle\max_{ \mathcal{A}\in \mathbb{A}}I(U;X|Y_{\mathcal{A}}) .
\end{align*} 
\label{thm2}
\end{prop}
\begin{proof}
    The proof is presented in Section \ref{sec:proof of theorem 1} and relies on an auxiliary result presented in Section \ref{sec:auxiliary result}.
\end{proof}
Building on the coding scheme of Proposition \ref{thm2}, we obtain the following result.
\begin{thm}
The coding scheme of Section \ref{part II} achieves the secret  rate
\begin{align*}\nonumber
  R_s=&\displaystyle{\max_{\substack{U,V\\\\ U-V-X-Y_{[J]}
}}}\left[\displaystyle{\min_{\mathcal{A}\in \mathbb{A}}}I({V;{Y_{\mathcal{A}}}|U})- {\displaystyle \max_{\mathcal{U}\in \mathbb{U}}}I(V;Y_{\mathcal{U}}|U)\right]^+\\ 
 & \phantom{}\text{subject to } R_p \geq \displaystyle\max_{ \mathcal{A}\in \mathbb{A}}I(U;X|Y_{\mathcal{A}})\!+\!\max_{ \mathcal{A}\in \mathbb{A}}I(V;X|UY_{\mathcal{A}}).
\end{align*}
\label{thm3}
\end{thm}
\begin{proof}
 See Section \ref{proof of theorem 2}. 
\end{proof}
Note that Theorem \ref{thm3} recovers Proposition \ref{thm2} by choosing $U=\emptyset$.

The achievable rates in Proposition \ref{thm2} and Theorem~\ref{thm3} could be obtained from \cite{tavangaran2016secret}. However, \cite{tavangaran2016secret} only provides an existence result and not an explicit coding scheme. 

In the case of rate-limited public communication,  converse results remain an open problem and are not addressed in \cite{tavangaran2016secret}. However, for rate-unlimited public communication, we have the following  corollary when all the participants  are needed to reconstruct the secret, and  
any strict subsets of participants in $[J]$ must not learn any information about the secret.
\begin{corollary} \label{cor1}
When $ R_p = +  \infty$, $\mathbb{A} = \{ [J] \}$, $\mathbb{U}= 2^{[J]}\backslash \mathbb{A }$, 
the coding scheme of Section \ref{seccodingscheme} achieves  the secret capacity  
\begin{align}\nonumber
\lim_{R_p \rightarrow {}{+\infty}} C_s(R_p)=\displaystyle{\min_{\mathcal{U}\subsetneq [J]}{I({X;{Y_{[J]}|Y_{\mathcal{U}}}}}}).
  \end{align}
\end{corollary}

  \begin{proof}
 The achievability follows from Proposition~\ref{thm2}, the Markov Chains  $X- Y_{[J]}- Y_{\mathcal{U}}$, $\mathcal{U} \in \mathbb{U}$, and because $\mathbb{U} = \{ \mathcal{S} \subseteq [J] : |\mathcal{S} |< J\}$. The converse follows from \cite{ahlswede1993common,maurer1993secret}.
  \end{proof}

\begin{ex} \label{ex1}
Consider $J = 2$ participants with the access structure
$\mathbb{A} = \{ \{1,2\} \}$ and $\mathbb{U} = \{ \{1\}, \{2\} \}$,
meaning that both participants are simultaneously needed to reconstruct the secret, while each participant alone learns nothing about it. This setup models a {two-factor authentication} scenario, where both factors must be combined to recover the secret. Furthermore, assume a binary symmetric test channel model for the source such that
$Y_1 = X \oplus N_1$, $Y_2 = X \oplus N_2$, 
where $\oplus$ denotes addition modulo 2, $X$ is a Bernoulli random variable with parameter $\frac{1}{2}$, and $N_1, N_2$ are independent Bernoulli random variables with parameters $p_1, p_2,$ respectively. We illustrate Theorem \ref{thm3} and Corollary~\ref{cor1} in Fig. \ref{fignum} for $p_1 = p_2 = 0.15$ and by choosing $V = \emptyset$, $U = X \oplus N_U$ with $N_U$ a Bernoulli random variable independent of all other random variables.
\end{ex}
                 \begin{figure}
\centering
 \includegraphics[width=8.9 cm]{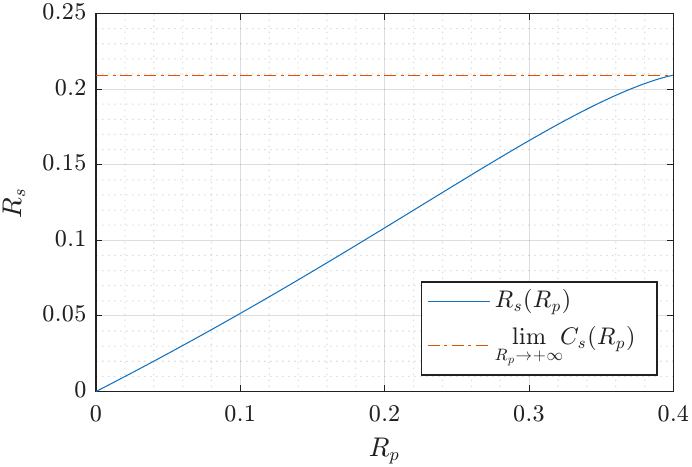}
 \caption{Illustration of the trade-off between public communication rate $R_p$ and  secret rate $R_s$, computed from Theorem~\ref{thm3} and Corollary~\ref{cor1}, for the  setup described in Example~\ref{ex1} with $J=2$, $\mathbb{A} =\{\{1,2\}\}$, and $\mathbb{U}= \{\{1\}, \{2\}\}$.
 }
    \label{fignum}
\end{figure}
\label{sec:main results}
\subsection{Application to Secret-Key Generation} \label{sec:key}
Our results yield an explicit and capacity-achieving coding scheme for secret-key generation under one-way rate-limited public communication between two parties in the presence or absence of an eavesdropper. The secret-key capacities in these settings have been established in~\cite{csiszar2000common} via non-constructive proofs and are reviewed~next.

\begin{thm}[{\cite[Th. 2.4]{csiszar2000common}}] \label{th5}
The one-way secret-key capacity under rate-limited public communication $R_p>0$ and in the absence of an eavesdropper is
\begin{align}\nonumber
   \max_{\substack{U\\   U-X-Y}}I(Y;U)   \text{ subject to } R_p \geq
   I(U;X) - I(U;Y)     .
\end{align}
\end{thm}

By Proposition \ref{thm2}, the coding scheme of Section \ref{seccodingscheme} achieves the secret-key capacity of Theorem  \ref{th5} with $J=1$ and $\mathbb{A}=\{ \{1\}\}$.

\begin{thm}[{\cite[Th. 2.6]{csiszar2000common}}]\label{th6}
The one-way secret-key capacity under rate-limited public communication $R_p>0$ and the presence of an eavesdropper  is
\begin{align}\nonumber
  & \max_{\substack{U,V\\   U-V-X-(Y,Z)}}[I(Y;V|U) -I(Z;V|U)]^+   \text{ subject to }\\\nonumber &  \phantom{-------}R_p \geq
    I(V;X) - I(V;Y)        .
\end{align}
\end{thm}
By Theorem \ref{thm3}, the coding scheme of Section \ref{seccodingscheme} achieves the secret-key capacity of Theorem  \ref{th6} with $J=2$, $\mathbb{A}=\{ \{1\}\}$, and $\mathbb{U}=\{ \{2\}\}$.

Note that \cite{chou2013polar} also provides a capacity-achieving and constructive scheme for secret-key generation under rate-limited communication. However,  \cite{chou2013polar} does not show the achievability of the capacity in Theorem \ref{th6} and, unlike this paper,  requires a pre-shared secret with negligible rate to ensure strong secrecy.
       \section{Auxiliary result} \label{sec:auxiliary result}
We provide an auxiliary result to construct an explicit coding scheme for the  problem described in Section~\ref{sec:problem statement}. Specifically, we consider the setting of Section \ref{sec:problem statement} with the following modifications. 
Instead of considering the constraints \eqref{eqnrel}-\eqref{eqn23}, the dealer creates a quantized version $\widetilde U^{1:N}$ of $X^{1:N}$, with the requirements that (i)  $\widetilde U^{1:N}$ follows a pre-determined product distribution, and (ii) any subsets of participants in the access structure can reconstruct $\widetilde U^{1:N}$. Next, we formalize this problem statement and construct an explicit coding scheme for this setting.  We will then use this auxiliary result to prove Proposition~$\ref{thm2}$ and Theorem~$\ref{thm3}$ in Sections~\ref{sec:proof of theorem 1} and \ref{proof of theorem 2}, respectively.

       Consider a discrete memoryless source  $((\mathcal X\times \mathcal Y_{\mathcal{A}} )_{\mathcal{A}\in \mathbb{A}},(p_{XY_{\mathcal{A}}})_{\mathcal{A}\in \mathbb{A}}) $ with $\mathcal{X} \triangleq \{ 0,1\}$. Define the joint probability distribution $p_{X_{}^{}U_{}^{}Y^{}_{{\mathcal{A}}}}\triangleq p_{X^{}{Y^{}_{{\mathcal{A}}}}}p_{U^{}|X^{}}$ over $\mathcal X\times \mathcal U\times \mathcal Y_{\mathcal{A}}$, ${\mathcal{A}} \in \mathbb{A}$, with $\mathcal{U} \triangleq \{ 0,1\}$. For any ${\mathcal{A}} \in  \mathbb{A}$, let $(X_{}^{1:N}, U_{}^{1:N}, Y_{{\mathcal{A}}}^{1:N})$ be distributed according to the product distribution $p_{X_{}^{1:N}U_{}^{1:N}Y^{1:N}_{{\mathcal{A}}}}\triangleq\prod_{i=1}^N p_{X_{}U{Y_{{\mathcal{A}}}}}$. 
              \begin{definition}
              A $(2^{NR_M},N)$ code consists of
              \begin{itemize}
                  \item An alphabet $\mathcal{M} \triangleq [ 2^{N R_M}]$; 
                  \item An encoding function $f : \mathcal X^{N} \rightarrow \mathcal{U}^{N}\times \mathcal{M}$; 
                      \item  $|\mathbb{A}|$ 
                      decoding functions $g_{\mathcal{A}} : \mathcal M \times \mathcal Y_{\mathcal{A}}^{N} \rightarrow \mathcal U^{N}  $,  $\mathcal{A}\in \mathbb{A}$;
                      \end{itemize}
                      and operates as follows:
                      \begin{enumerate}
                      \item  The dealer observes
                      $ X^{1:N}$ and  Participant $j \in [J]$ observes $Y_j^{1:N}$;
                         \item              The dealer encodes $X_{}^{1:N}$ with $f$ to form a vector quantized version $\widetilde U_{}^{1:N}$ of $X_{}^{1:N}$ and a message $M$, where $(\widetilde U_{}^{1:N},M) \triangleq f(X^{1:N})$.  Let  the  probability distribution of $(\widetilde U^{1:N}, X^{1:N})$ be denoted by $\widetilde  p_{U^{1:N}_{}X^{1:N}_{}}$;
                      \item   The dealer transmits $M$ to the participants over the public channel; 
    
                                         \item Any subset of  participants ${\mathcal{A}}\in \mathbb{A}$ can create  $\hat{U}^{1:N}_{\mathcal{A}}\triangleq g_{\mathcal{A}}(M, Y_{{\mathcal{A}}}^{1:N})$, an estimate of $\widetilde U_{}^{1:N}$.
                    \end{enumerate}    \label{dfn1}         
                                 \end{definition}
      In the following, we assume that the participants and dealer observe $k$ blocks of lengths $N$, i.e., 
      $Y^{1:N}_{{\mathcal{A}},1:k} \triangleq (Y^{1:N}_{{\mathcal{A}},i})_{i \in [k]}$ is the source observation  at Decoder ${\mathcal{A}}\in \mathbb{A}$ and  $X^{1:N}_{1:k} \triangleq (X^{1:N}_{i})_{i \in [k] }$ is the source observation at the dealer  with  $p_{X^{1:N}_{1:k}Y^{1:N}_{{\mathcal{A}},1:k}} \triangleq  \prod_{i=1}^{kN} p_{XY_{\mathcal{A}}}$, ${\mathcal{A}}\in\mathbb{A}$.
                 
                \begin{thm}   For any $\epsilon, \delta >0$, there exist $k,n_0 \in \mathbb{N}$ such that for any $n\geq n_0$ and $N \triangleq 2^n$,          there exists an encoder
     \begin{align}
         &f      : \mathcal{X}^{kN}_{}\to \mathcal{U}^{kN}_{}\times \mathcal{M}, \label{encoder}
     \end{align}
          and $|\mathbb{A}|$ decoders
      \begin{align*}
       g_{\mathcal{A}}
       : \mathcal M 
       \times \mathcal{Y}_{{\mathcal{A}}}^{kN} \to \mathcal{U}^{kN}, \forall {\mathcal{A}}\in \mathbb{A}, 
     \end{align*}
such that 
    the message rate $R_{{M}}$ satisfies
      \begin{align}
        R_{{M}} = \displaystyle\max_{ {\mathcal{A}} \in \mathbb{A}}I(U;X|Y_{\mathcal{A}})+ \delta,
        \label{eqn2}
     \end{align} and the probability of error at Decoder ${\mathcal{A}}\in \mathbb{A}$ satisfies
      \begin{align}
       \mathbb{P} [  g_{\mathcal{A}}
      (M, {Y}^{1:N}_{ \mathcal{A},{1:k}})\neq\widetilde{U}^{1:N}_{1:k}]<\epsilon, \forall  {\mathcal{A}}\in \mathbb{A},\label{eqn3}
     \end{align}
where       $(\widetilde{U}^{1:N}_{1:k}, M) \triangleq f (X^{1:N}_{1:k})$.  Furthermore, the  encoding and decoding complexities are 
${O}(kN \log N)$, since the scheme operates over $k$ blocks of length $N$, 
each of which inherits the ${O}(N \log N)$ complexity of 
source polarization~\cite{arikan2010source}.

 Additionally, we have the following two properties.
\begin{itemize}
    \item

 The probability distribution induced by the encoding scheme $\widetilde p_{U^{1:N}_{i}X^{1:N}_{i}}$,  $i\in [k],$  satisfies   
  \begin{align}
     \mathbb{V} \left( \widetilde{p}_{U_{i}^{1:N}X_{i}^{1:N}} , p_{U^{1:N}_{}X^{1:N}_{}}   \right) \leq \delta_N^{(1)}  ,\label{eqn4}     \end{align}

 where ${p}_{U_{}^{1:N}X_{}^{1:N}} \triangleq  \prod_{i=1}^{N} p_{UX}$, $
 \delta_N \triangleq 2^{-N^{\beta}},
 \beta \in]0,\textstyle\frac{1}{2}[$, and $\delta_N^{(1)} \triangleq \sqrt{2 \ln 2} \sqrt{N\delta_N}$.   

\item For any $\mathcal{U} \in \mathbb{U}$,
    \begin{align}
     &   H(\widetilde U^{1:N}_{1:k} | Y^{1:N}_{\mathcal{U},1:k}) \nonumber \\
     & \phantom{--}\geq Nk (H(U|Y_{\mathcal{U}}) -  H(U|X)) - o(Nk). \label{eqlemma2}
    \end{align}
 \end{itemize}
     \label{thm1}
     \end{thm}

Establishing Proposition~\ref{thm2} requires combining the result of Theorem~\ref{thm1} with universal hashing. To this end, the property in~\eqref{eqn4} enables the analysis of the induced distribution within each coding block, showing that it is close to a product distribution. However, a key challenge is that the joint distribution across the $k$ encoding blocks is not a product distribution, in particular, $(\widetilde{U}^{1:N}_{1:k}, Y^{1:N}_{\mathcal{U},1:k})$ are not independent and identically distributed random variables. Despite this complication, the property in~\eqref{eqlemma2} provides a lower bound on the conditional entropy $H(\widetilde{U}^{1:N}_{1:k} \mid Y^{1:N}_{\mathcal{U},1:k})$, which will play a crucial role in the analysis of universal hashing.

 Note that the property in \eqref{eqn4} generalizes the lossless source coding problem where the encoder output distribution must approximate a target distribution~\cite{chou2013data,chou2016coding,chou2021universal}, to the case of lossy source coding with compound side information.

\section{Proof of Theorem \ref{thm1}} \label{sec:proof of theorem 3}

We begin by outlining the high-level idea for the case $\mathbb{A} = \{\mathcal{A}_1,\mathcal{A}_2 \}$, which will later be extended by induction over the cardinality of $\mathbb{A}$.

\emph{Encoding}. We operate over $k$ blocks, each of length $N$. 
For each block $j \in \llbracket 1, k \rrbracket$, two encoders $f_1$ and $f_2$ encode the source sequence $X_j^{1:N}$ as 
$$
f_i(X_j^{1:N}) = (\widetilde{U}_{j}^{1:N}, M_{i,j}), \quad i \in \{1,2\},
$$
 such that $\widetilde{U}_{j}^{1:N}$ can be reconstructed from the side information $Y_{\mathcal{A}_i,j}^{1:N}$ and $M_{i,j}$. The transmitted messages are
\begin{align}
[\, M_{1,1}, \, (M_{1,i+1} \oplus M_{2,i})_{i=1}^{k-1}, \, M_{2,k} \,]. \label{eqchaining}
\end{align}
Then, each decoder can successively reconstruct the entire sequence  
$\widetilde{U}_{1:k}^{1:N}$ using only its own side information and the messages defined in \eqref{eqchaining}. Specifically, 
Decoder~1, with $Y_{\mathcal{A}_1,1:k}^{1:N}$, begins decoding from the first block and proceeds sequentially toward the last,
while Decoder~2, with $Y_{\mathcal{A}_2,1:k}^{1:N}$, starts from the last block and proceeds in reverse order toward the first, as described next.

\emph{Decoding for $\mathcal{A}_1$}. Suppose a receiver corresponding to $\mathcal{A}_1$, has access to side information 
$Y_{\mathcal{A}_1,1:k}^{1:N}$. 
Then, decoding proceeds   as follows:  Using $(M_{1,1}, Y_{\mathcal{A}_1,1}^{1:N})$, 
    the decoder reconstructs $\widetilde{U}_{1}^{1:N}$,
    and from it computes $M_{2,1}$.
Then, for $i = 1, \ldots, k-1$, use the previously recovered $M_{2,i}$ to obtain
    $$
    M_{1,i+1} = (M_{1,i+1} \oplus M_{2,i}) \oplus M_{2,i}.
    $$
    Then, using $(M_{1,i+1}, Y_{\mathcal{A}_1,i+1}^{1:N})$, 
    reconstruct $\widetilde{U}_{i+1}^{1:N}$ and from it computes $M_{2,i+1}$.
 
\emph{Decoding for $\mathcal{A}_2$}.
A receiver corresponding to $\mathcal{A}_2$, 
with side information $Y_{\mathcal{A}_2,1:k}^{1:N}$, proceeds backward: 
    Using $(M_{2,k}, Y_{\mathcal{A}_2,k}^{1:N})$, reconstruct $\widetilde{U}_{k}^{1:N}$ and compute $M_{1,k}$.
    Then, for $i = k-1, \ldots, 1$, recover
    $$
    M_{2,i} = (M_{1,i+1} \oplus M_{2,i}) \oplus M_{1,i+1},
    $$
    and decode $\widetilde{U}_{i}^{1:N}$ using 
    $(M_{2,i}, Y_{\mathcal{A}_2,i}^{1:N})$.

 \subsection{
     Notation} \label{secnotation}
Let $N \triangleq 2^n$, $n \in \mathbb N$ , and 
$G_n \triangleq  \left[ \begin{smallmatrix}
       1 & 0            \\[0.3em]
       1 & 1 
     \end{smallmatrix} \right]^{\otimes n} $ be the source polarization matrix defined in \cite{arikan2010source}.    
      Define 
     $V_{}^{1:N} \triangleq U_{}^{1:N} G_n$. 
     For  $\delta_N \triangleq 2^{-N^{\beta}}$, $\beta \in]0,\frac{1}{2}[$, define  for any $\mathcal{A}\in \mathbb{A}$,
\begin{align}
    \mathcal{V}_U &\triangleq \left\{ i \in [N]: H( V^i | V^{1:i-1})\geq  1 - \delta_N \right\},\label{eqnvU}\\
      \mathcal{H}_U &\triangleq \left\{ i \in [N]: H( V^i | V^{1:i-1}) \geq   \delta_N \right\},\label{eqnhu}\\
          \mathcal{V}_{U|X} &\triangleq \left\{ i \in [N]: H( V^i | V^{1:i-1}X^{1:N}) \geq  1-\delta_N \right\},\nonumber
\\
\mathcal{H}_{U|Y_{\mathcal{A}}} &\triangleq \left\{ i \in [N]: H( V^i | V^{1:i-1}Y_{\mathcal{A}}^{1:N}) \geq  \delta_N \right\}.\label{eqnhuya}
\end{align}
         In the coding scheme analysis, we will use Lemma \ref{lem1}. Note that Lemma \ref{lem1} also gives an interpretation for the sets defined in $(\ref{eqnhu})$,  $(\ref{eqnhuya})$ in terms of source coding with side information. 
  \begin{lem}(Source coding with side information\cite{arikan2010source}). Consider a probability distribution $p_{UY}$ over $\mathcal{U}\times \mathcal{Y}$ with $\mathcal{U} \triangleq \{ 0,1\}$ and $\mathcal{Y}$ a finite alphabet. 
       Consider $(U^{1:N}, Y^{1:N} )$ distributed according to $\prod_{i=1}^N p_{UY}$. Define $A^{1:N}\triangleq U^{1:N}G_n$ and for $\delta_N\triangleq 2^{-N^\beta}$ with $\beta \in]0,\frac{1}{2}[$, the set $  \mathcal{H}_{U|Y} \triangleq \left\{ i \in [N]        : H( A^i | A^{1:i-1}Y^{1:N}) > \delta_N \right\}$. Given $A^{1:N}[ \mathcal{H}_{U|Y}]$ and $Y^{1:N}$, one can form $\hat A^{1:N}$ by the successive cancellation decoder of \cite{arikan2010source} such that $\mathbb{P}[\hat A^{1:N}\neq A^{1:N}]\leq N\delta_N$. Moreover, $\displaystyle \lim_{N \to \infty}|\mathcal{H}_{U|Y}|/N=H(U|Y)$.\label{lem1}
     \end{lem}
We write the access structure as $\mathbb{A} = \{ \mathcal{A}_{j} \subseteq [J] : j \in [|\mathbb{A}| ] \}$.      Because our coding scheme operates over $k$ blocks of $N$ source observations $(X_{1:k}^{1:N}, Y_{\mathcal{A}_j,1:k}^{1:N})$,  to simplify notation, we define $B_{j} \triangleq Y_{\mathcal{A}_j}$ and $B_{j,1:k}^{1:N} \triangleq Y_{\mathcal{A}_j,1:k}^{1:N}$. 
    This notation is also convenient as the analysis of our coding scheme is done by induction over the number of decoders, i.e., \textcolor{black}{the cardinality of $\mathbb{A}$}.
Then, for a finite set  of integers~$\mathcal{S} \subseteq [|\mathbb{A}|]$, consider the following notation for a set of $|\mathcal{S}|$ decoders $(g_{\mathcal{S},j}^{1:N})_{j\in\mathcal{S}}$. The superscript indicates the length of the output, the first subscript $\mathcal{S}$ indicates that all the decoders are indexed by $\mathcal{S}$, the second subscript $j\in\mathcal{S}$ indicates that $g_{\mathcal{S},j}^{1:N}$ is the decoding function for the decoder indexed by $j \in \mathcal S$, i.e., the decoder that has access to the source observations $B_{j,1:k}^{1:N}$. 
Consider also the following notation for an encoder $f_{\mathcal{S}}^{1:N}$. The superscript indicates the length of the input, and the subscript $\mathcal{S}$ indicates that $f_{\mathcal{S}}^{1:N}$ is associated with $|\mathcal{S}|$ decoders that are indexed by $\mathcal{S}$.
    \subsection{Preliminary Results}
Algorithm $\ref{algm1}$ describes the  construction of $\widetilde U^{1:N}_{1:k}$ from the random variables $X^{1:N}_{1:k}$ and $R_{1}$, where $R_{1}$ is a vector of $|\mathcal{V}_{U|X}|$ uniformly distributed bits.
     \begin{algorithm}
  \caption{Construction of $\widetilde U^{1:N}_{1:k}$ 
  }
  \label{alg:encoding_6}
  \begin{algorithmic}   [1] 
    \REQUIRE Source observations $(X_{i}^{1:N})_{i \in [k]}$, where  $X_{i}^{1:N}$ is of length $N$ and corresponds to Block~$i\in [k]$; a sequence $R_{1}$ of $|\mathcal{V}_{U|X}|$ uniformly distributed bits and $k$ sequences $(\bar{R}_i)_{i \in [k]}$ of $|\mathcal{V}_{U} \backslash \mathcal{V}_{U\vert X}|$ uniformly distributed bits.
            \FOR{Block $i=1$ to $k$}
    \STATE $\widetilde V_{i}^{1:N}[\mathcal{V}_{U\vert X}]\leftarrow R_{1}$
        \STATE $\widetilde V_{i}^{1:N}[\mathcal{V}_{U} \backslash \mathcal{V}_{U\vert X}]\leftarrow \bar R_{i}$
       \STATE Given $X_{i}^{1:N}$, successively draw the remaining bits of $\widetilde V_{i}^{1:N}$ according to\\$\widetilde p_{V_{i}^{1:N}X_{i}^{1:N}}\triangleq  \prod_{j=1}^N \widetilde p_{V_{i}^j\vert V_{i}^{1:j-1}X_{i}^{1:N} }p_{X^{1:N}} $   
       with
  \begin{align}  
\nonumber & \widetilde p_{V_{i}^j\vert V_{i}^{1:j-1} X_{i}^{1:N}}(v_{i}^j\vert \widetilde V_{i}^{1:j-1}X_{i}^{1:N})
\\ & \triangleq p_{V^j\vert V^{1:j-1} }(v_{i}^j\lvert \widetilde V_{i}^{1:j-1}X_{i}^{1:N})  
   \text{ if } j \!\in\! \mathcal{V}_{U}^c.
        \label{eqnini}
\end{align}
 \STATE Construct $\widetilde U_{i}^{1:N} \triangleq \widetilde V_{i}^{1:N} G_n$
  \ENDFOR
   \RETURN $\widetilde U^{1:N}_{1:k} \triangleq (\widetilde U^{1:N}_{i})_{i \in [k]}$ 
           \end{algorithmic}\label{algm1}
\end{algorithm}
We now prove that \eqref{eqn4} and~\eqref{eqlemma2} hold.

\subsubsection{Proof of \eqref{eqn4}}
 We show that  the product distribution ${p}_{U_{}^{1:N}X_{}^{1:N}}$
 is well approximated in each Block $i\in [k]$. 
Specifically, for  $i\in [k] $, we~have
         \begin{align}\nonumber
& \mathbb{D}(p_{X^{1:N}U^{1:N}}\lVert\widetilde p_{X_{i}^{1:N}U_{i}^{1:N}})\\\nonumber
&\stackrel{(a)}=\mathbb{D}(p_{X^{1:N}V^{1:N}}\lVert\widetilde p_{X_{i}^{1:N}V_{i}^{1:N}})\\\nonumber
&\stackrel{(b)}=\mathbb{D}(p_{V^{1:N}|X^{1:N}}\lVert\widetilde p_{V_{i}^{1:N}|X_{i}^{1:N}}) 
\\\nonumber
&\stackrel{(c)}=\sum_{j=1}^N\mathbb{D}(p_{V^{j}|V^{1:j-1}X^{1:N}}\lVert\widetilde p_{V_{i}^{j}|V_{i}^{1:j-1}X_{i}^{1:N}}) 
\\\nonumber
  &\stackrel{(d)}=\sum_{j \in \mathcal{V}_{U}}(1- H({V^j \vert V^{1:j-1}}X^{1:N})) \displaybreak[0]\\\nonumber
 & \stackrel{(e)}\leq |\mathcal{V}_{U}|\delta_N
 \\
 & \leq N\delta_N, \label{eqn5}
 \end{align} where $(a)$ holds by the invertibility of $G_n$, $(b)$ and $(c)$ hold by the chain rule for relative entropy, $(d)$ holds by (\ref{eqnini}) and by Line~$3$ of Algorithm~\ref{algm1}, $(e)$ holds by  $(\ref{eqnvU})$. Finally, we deduce \eqref{eqn4} from Pinkser's inequality and \eqref{eqn5}.
 \subsubsection{Proof of \eqref{eqlemma2}}
To prove \eqref{eqlemma2}, we consider the construction of the following auxiliary random variables $\breve U^{1:N}_{1:k}$  in Algorithm~\ref{alg:encoding_6b}. Note that the difference in the construction of $\breve U^{1:N}_{1:k}$ and $\widetilde U^{1:N}_{1:k}$ is that the randomness in Line 3 of Algorithm~\ref{alg:encoding_6b} is reused in each block for the construction of $\widetilde U^{1:N}_{1:k}$ in Algorithm~\ref{alg:encoding_6}. Let $\breve{p}_{U_{i}^{1:N}X_{i}^{1:N}} $ be the joint distribution between of $(\breve U^{1:N}_{1:k},X^{1:N}_{1:k})$, then, similar to the proof of \eqref{eqn4} and \cite[Lemma 2]{chou2015using}, we have
 \begin{align}
     \mathbb{V} \left( \breve{p}_{U_{i}^{1:N}X_{i}^{1:N}} , p_{U^{1:N}_{}X^{1:N}_{}}   \right) \leq  \delta_N^{(1)} +\delta_N^{(2)},   \label{eqdistc}  \end{align}
     where $\delta_N^{(2)} \triangleq N \sqrt{\delta_N + 2 \delta_{N}^{(1)} (N - \log \delta_{N}^{(1)})}$.
  \begin{algorithm}
  \caption{Construction of $\breve U^{1:N}_{1:k}$ 
  }
  \label{alg:encoding_6b}
  \begin{algorithmic}   [1] 
    \REQUIRE Source observations $(X_{i}^{1:N})_{i \in [k]}$, where  $X_{i}^{1:N}$ is of length $N$ and corresponds to Block~$i\in [k]$; $k$ sequences $(\check {R}_i)_{i \in [k]}$  and $k$ sequences $(\bar{R}_i)_{i \in [k]}$ of  $|\mathcal{V}_{U} |$ 
 and $|\mathcal{V}_{U} \backslash \mathcal{V}_{U\vert X}|$, respectively,  uniformly distributed~bits.
            \FOR{Block $i=1$ to $k$}
        \STATE $\breve V_{i}^{1:N}[\mathcal{V}_{U\vert X}]\leftarrow \check R_{i}$
        \STATE $\breve V_{i}^{1:N}[\mathcal{V}_{U} \backslash \mathcal{V}_{U\vert X}]\leftarrow \bar R_{i}$

       \STATE Given $X_{i}^{1:N}$, successively draw the remaining bits of $\breve V_{i}^{1:N}$ according to\\$\breve p_{V_{i}^{1:N}X_{i}^{1:N}}\triangleq  \prod_{j=1}^N \breve p_{V_{i}^j\vert V_{i}^{1:j-1}X_{i}^{1:N} }p_{X^{1:N}} $   
       with
  \begin{align*}  
\!\!\!\!\!\!\!\!\!\!\!\!\!\!\!\!\!\!\!\!\!\!\!\!\!\!\!\!\!\!\!\!\!\!\!\!\!\!\!\!\!\!\!\!\!\!\!\! \breve p_{V_{i}^j\vert V_{i}^{1:j-1} X_{i}^{1:N}}(v_{i}^j\vert \breve V_{i}^{1:j-1}X_{i}^{1:N})  
\end{align*}
\begin{align}  \label{eqdeter}\triangleq\begin{cases}p_{V^j\vert V^{1:j-1}X^{1:N} }(v_{i}^j\lvert \breve V_{i}^{1:j-1}X_{i}^{1:N})     \text{ if } j\! \in \! \mathcal{V}^c_{U}\backslash\mathcal{H}^c_{U} , \\\mathds{1}\{v_j \!=\! \displaystyle\argmax_v p_{V^j|V^{1:j-1}} (v^j|v^{1:j-1}) \} \text{ if } j \!\in\! \mathcal{H}_{U}^c.         \end{cases}\end{align}

 \STATE Construct $\breve U_{i}^{1:N} \triangleq \breve V_{i}^{1:N} G_n$
  \ENDFOR
   \RETURN $\breve U^{1:N}_{1:k} \triangleq (\breve U^{1:N}_{i})_{i \in [k]}$ 
           \end{algorithmic}
\end{algorithm}

Then, we have
\begin{align*}
    &H(\widetilde U^{1:N}_{1:k} | Y^{1:N}_{\mathcal{U},1:k}) \\
    & \stackrel{(a)}\geq H(\widetilde U^{1:N}_{1:k} | Y^{1:N}_{\mathcal{U},1:k} R_1)\\
    & \stackrel{(b)}= \sum_{i=1}^k  H(\widetilde U^{1:N}_{i} | Y^{1:N}_{\mathcal{U},1:k} \widetilde U^{1:N}_{1:i-1} R_1)\\
    & \stackrel{(c)}= \sum_{i=1}^k  H(\widetilde U^{1:N}_{i} | Y^{1:N}_{\mathcal{U},1:k}  R_1)\\
    & = \sum_{i=1}^k  H(\widetilde U^{1:N}_{i} | Y^{1:N}_{\mathcal{U},i}  R_1)\\ 
    & \stackrel{(d)}= \sum_{i=1}^k  [H(\widetilde U^{1:N}_{i} \breve U^{1:N}_{i} | Y^{1:N}_{\mathcal{U},i}  R_1)- H( \breve U^{1:N}_{i} | Y^{1:N}_{\mathcal{U},i} \widetilde U^{1:N}_{i}  R_1)]\\
    & \geq \sum_{i=1}^k  [H( \breve U^{1:N}_{i} | Y^{1:N}_{\mathcal{U},i}  R_1)- H( \breve U^{1:N}_{i} | Y^{1:N}_{\mathcal{U},i} \widetilde U^{1:N}_{i}  R_1)]\\
    & \stackrel{(e)}= \sum_{i=1}^k  [H( \breve U^{1:N}_{i} | Y^{1:N}_{\mathcal{U},i}  )- H( \breve U^{1:N}_{i} | Y^{1:N}_{\mathcal{U},i} \widetilde U^{1:N}_{i}  R_1)]\\
    & \stackrel{(f)}= \sum_{i=1}^k  [H( \breve U^{1:N}_{i} | Y^{1:N}_{\mathcal{U},i}  ) \\
    & \phantom{--}- H( \breve V^{1:N}_{i}[\mathcal{V}_{U|X} \cup (\mathcal{V}^c_{U} \backslash \mathcal{H}^c_{U})] | Y^{1:N}_{\mathcal{U},i} \widetilde U^{1:N}_{i}  R_1)]\\
    & \stackrel{(g)}\geq \sum_{i=1}^k  [H( \breve U^{1:N}_{i} | Y^{1:N}_{\mathcal{U},i}  )- |\mathcal{V}_{U|X} | - (|\mathcal{V}^c_{U}| - |\mathcal{H}^c_{U}| )]\\
    & \stackrel{(h)}= \sum_{i=1}^k  [H( \breve U^{1:N}_{i} | Y^{1:N}_{\mathcal{U},i}  )- |\mathcal{V}_{U|X} | - o(N)]\\
    & \stackrel{(i)}= \sum_{i=1}^k  [N H(U | Y_{\mathcal{U}}  )- |\mathcal{V}_{U|X} | - o(N)]\\
    & \stackrel{(j)}= \sum_{i=1}^k  [N H(U | Y_{\mathcal{U}}  )- NH(U|X) - o(N)]\\
    &= Nk( H(U | Y_{\mathcal{U}}  )- H(U|X)) - o(kN),
\end{align*}
where $(a)$ holds because conditioning reduces entropy, $(b)$ and $(d)$ hold by the chain rule, $(c)$ holds by the Markov chain $ \widetilde U^{1:N}_{i} - (Y^{1:N}_{\mathcal{U},1:k} ,R_1) - \widetilde U^{1:N}_{1:i-1} $, $(e)$ holds because $R_1$ is independent of $(\breve U^{1:N}_{i} , Y^{1:N}_{\mathcal{U},i})$, $(f)$ holds by invertibility of $G_n$,  because $\breve V^{1:N}_{i}[\mathcal{V}_{U} \backslash \mathcal{V}_{U|X} ] = \widetilde V^{1:N}_{i}[\mathcal{V}_{U} \backslash \mathcal{V}_{U|X} ]$, and because $\breve V^{1:N}_{i}[\mathcal{H}^c_{U}]$ can be deterministically computed from $\breve V^{1:N}_{i}[\mathcal{H}_{U}]$ as seen in \eqref{eqdeter}, $(g)$ holds because $\mathcal{V}_{U} \subseteq \mathcal{H}_{U}$, $(h)$ holds by \cite[Lemma 1]{chou2013polar} and \cite{arikan2010source}, $(i)$ holds by \cite{csiszar1996almost} and \eqref{eqdistc},  $(j)$ holds by \cite[Lemma 1]{chou2013polar}.

   \subsection{Proof of Equations \texorpdfstring {$(\ref{eqn2})$}{Lg} and \texorpdfstring{$(\ref{eqn3})$}{Lg}}\label{subsec:proof}

This proof is by induction. 
Assume that $|\mathbb{A}|=1$. 
       The encoder and decoder for the case $|\mathbb{A}|=1$ are defined  in Algorithms $\ref{algm2}$ and $\ref{algm3}$, respectively.

      \begin{algorithm}
  \caption{Encoding when $|\mathbb{A}|=1$ }
  \label{alg:encoding_1}
  \begin{algorithmic}   [1] 
    \REQUIRE Source observations $X^{1:N}_{1:k} \triangleq (X_{i}^{1:N})_{i \in [k]}$, where  $X_{i}^{1:N}$ is of length $N$ and corresponds to Block~$i\in [k]$; a sequence $R_{1}$ of $|\mathcal{V}_{U|X}|$ uniformly distributed bits.
    \STATE Transmit $R_{1}$ to Bob over the public channel  
    \STATE Construct $\widetilde{U}_{1:k}^{1:N}$ using Algorithm $\ref{algm1}$ with input the random variables $X^{1:N}_{1:k}$ and $R_{1}$ 
   \STATE  For $i \in [k]$ define \begin{align}
            M_{i}\triangleq 
      (\widetilde{V}_{i}^{1:N}[ \mathcal{H}_{U|B_1}\backslash \mathcal{V}_{U|X}]).\label{eqnen1}
     \end{align}
   
     \STATE Transmit $M_{i}$  to Bob over the public channel
       \STATE Define $ f^{1:N}_{\{1\}}( X^{1:N}_{1:k})\triangleq (\widetilde{U}_{1:k}^{1:N} ,M_{1:k})$
    \end{algorithmic}\label{algm2}
\end{algorithm}
\begin{algorithm}
  \caption{Decoding when $|\mathbb{A}|=1$}
  \label{alg:dencoding}
  \begin{algorithmic}   [1] 
    \REQUIRE Source observations $B^{1:N}_{1,1:k} \triangleq (B_{1,i}^{1:N})_{i \in [k]}$, where  $B_{1,i}^{1:N}$ is of length $N$ and corresponds to Block~$i\in [ k]$;
    the messages $M_{1:k}$ and $R_{1}$. 
 \FOR{Block $i=1$ to $k$} 
   \STATE Form $\widetilde{V}_{i}^{1:N}[ \mathcal{H}_{U|B_1}]$ from $M_i$ and $R_{1}$
 \STATE Create $\hat{U}_{1,i}^{1:N}$, an estimate of $\widetilde{U}_{i}^{1:N}$, from $( \widetilde{V}_{i}^{1:N}[ \mathcal{H}_{U|B_1}], B_{1,i}^{1:N})$ by Lemma \ref{lem1}.
     \ENDFOR
      \STATE Define $\hat{U}_{1,1:k}^{1:N} \triangleq (\hat{U}_{1,i}^{1:N})_{i\in[k]}$
     \STATE Define $g_{\{1\},1}^{1:N}(M_{1:k}, B^{1:N}_{1,1:k}) \triangleq \hat{U}_{1,1:k}^{1:N} $
               \end{algorithmic}\label{algm3}
\end{algorithm}
             
The communication rate is the rate of $M_i$, $i \in [k]$, plus the rate of the randomness $R_{1}$, i.e.,
      \begin{align*}
    R_{\{1\}}   &\stackrel{(a)}=\frac{k|\mathcal{H}_{U|B_1}\backslash \mathcal{V}_{U|X}|+|\mathcal{V}_{U|X}|}{kN}
   \\
     & \xrightarrow{N \to +\infty}\frac{k(H(U|B_1)-H(U|X))+H(U|X)
     }{k}\\
     &\leq H(U|B_1)-H(U|X) + \frac{H(U|X)}{k}   \\ & \xrightarrow{k \to \infty} H(U|B_1)-H(U|X)\\
     &\stackrel{(b)} =I(U,X|B_1),
        \end{align*} where $(a)$ holds by the definition of $M_i$ for $i \in [k]$ (Algorithm~$\ref{algm2}$, Line $3$), and the first limit holds by Lemma~$\ref{lem1}$ and \cite[Lemma 1]{chou2013polar} and because $U-X-B_1$ forms a Markov chain so that $\mathcal{V}_{U|X}\subset \mathcal{H}_{U|B_1}$, $(b)$ holds by the Markov Chain $U-X-B_1$.

Then, we study the probability of error as follows.  For $i \in [k]$, consider a coupling \cite{aldous1983random} between $\widetilde p_{ U_{i}^{1:N}}$ and $p_{ U_{}^{1:N}}$ such that $\mathbb{P}[\mathcal{E}_{i}]=\mathbb{V}(\widetilde p_{ U_{i}^{1:N}},p_{ U_{}^{1:N}} )$, where $\mathcal{E}_{i}\triangleq \{\widetilde U_{i}^{1:N}\neq U_{}^{1:N}\}$.          
For $i \in [k]$, we have 
   \begin{align}\nonumber
   & \mathbb{P}[\hat U_{1, i}^{1:N}\neq \widetilde U_{i}^{1:N}]\\\nonumber
   & = \mathbb{P}[\hat U_{1, i}^{1:N}\neq \widetilde U_{i}^{1:N}|\mathcal{E}_{i_{}}]\mathbb{P}[\mathcal{E}_{i_{}}]+\mathbb{P}[\hat U_{1, i}^{1:N}\neq \widetilde U_{i}^{1:N}|\mathcal{E}^c_{i_{}}]\mathbb{P}[\mathcal{E}^c_{i_{}}]\\\nonumber
       & \leq \mathbb{P}[\mathcal{E}_{i_{}} ]+\mathbb{P}[\hat U_{1, i}^{1:N}\neq \widetilde U_{i}^{1:N}|\mathcal{E}_{i_{}}^c ]\\\nonumber
              &=\mathbb{V}(\widetilde p_{ U_{i}^{1:N}},p_{ U_{}^{1:N}} )+\mathbb{P}[\hat U_{1, i}^{1:N}\neq \widetilde U_{i}^{1:N}|\mathcal{E}_{i_{}}^c]\\\nonumber
               &\leq \mathbb{V}(\widetilde p_{X_{i}^{1:N} U_{i}^{1:N}},p_{X_{}^{1:N} U_{}^{1:N}} )+\mathbb{P}[\hat U_{1, i}^{1:N}\neq \widetilde U_{i}^{1:N}|\mathcal{E}_{i_{}}^c]\\\nonumber
               &\stackrel{(a)}\leq \sqrt{2 \ln 2}\sqrt{N\delta_N^{}}+\mathbb{P}[\hat U_{1, i}^{1:N}\neq \widetilde U_{i}^{1:N}|\mathcal{E}_{i_{}}^c]\\
                  &\stackrel{(b)}\leq \sqrt{2 \ln 2}\sqrt{N\delta_N^{}}+N\delta_N,
                  \label{eqnerror}       \end{align}
             where $(a)$ holds by  $(\ref{eqn5})$ and Pinsker's inequality, 
              $(b)$ holds because $\mathbb{P}[\hat U_{1, i}^{1:N}\neq \widetilde U_{i}^{1:N}|\mathcal{E}_{i_{}}^c]\leq N\delta_N$, $\delta_N\triangleq 2^{-N^{\beta}}$, $\beta \in]0,\frac{1}{2}[, $  by \cite{arikan2010source}.  
      Finally, 
       \begin{align*}
           \mathbb{P}[\hat U_{1, 1:k}^{1:N}\neq \widetilde U_{1:k}^{1:N}] & \stackrel{(a)} \leq \displaystyle\sum_{{i=1}}^{{k}} \mathbb{P}[\hat U_{1,i}^{1:N}\neq \widetilde U_{i}^{1:N}]\\
           &\stackrel{(b)} \leq k(\sqrt{2 \log 2}\sqrt{N\delta_N^{}}+N\delta_N),
           \end{align*}
           where $(a)$ holds by the union bound, $(b)$ holds by $(\ref{eqnerror})$.        This proves the theorem for $|\mathbb{A}|=1$. 
           
           Now suppose that the theorem holds for $|\mathbb{A}|=L$. 
      Fix $\epsilon>0$ and $\delta>0$. By the induction hypothesis,
      there exist $k_1,n_1$ such that for any $n\geq n_1$, there is an encoder
      \begin{align}
         f_{[L]}^{1:k_1N}: \mathcal{X}^{k_1N}\to \mathcal{U}^{k_1N} \times \mathcal{M},\label{eqnfL}
     \end{align}and $L$ decoders
      \begin{align*}
         g_{[L], l}^{1:k_1N}: \mathcal{M}\times \mathcal{B}_{l}^{k_1N} \to \mathcal{U}^{k_1N}, \text{ for }  l\in [L],
     \end{align*}such that  the communication rate  is
      \begin{align}
\displaystyle\max_{ l\in [L]} I(U;X|B_l)+\frac{\delta}{2}\label{eqn11},
     \end{align}
and, for $k_2$ large enough,\footnote{Specifically, choose $k_2$ large enough such that $$\left[\!1\!\!+\! \frac{1}{k_2}\right] \!\!\! \left[ \displaystyle\max_{ l\in [L+1]}\!\! I(U;X|B_l)+\frac{\delta}{2} \right]\! +\!\frac{H(U|X)}{k_2} \leq \displaystyle \max_{ l\in [L+1]} \!\!I(U;X|B_l)\!+\!\delta.$$} the probability of error satisfies
      \begin{align}
        \mathbb{P} [  g_{[L], l}^{1:k_1N}(M_{[L]}, {B}_{l,1:k_1}^{1:N} ) \neq {\widetilde U}^{1:N}_{1:k_1}]\leq \frac{\epsilon}{k_2},\forall l\in [L]\label{eqn12},
     \end{align}
     where $(\widetilde{U}_{1:k_1}^{1:N},M_{[L]}) \triangleq f_{[L]}^{1:k_1N} ({X}_{1:k_1}^{1:N}) $.

Next, as in the case $|\mathbb{A}|=1$, define the encoder
    \begin{align}
      f^{1:k_1N}_{\{L+1\}} :  \mathcal{X}^{k_1N}\to \mathcal{U}^{k_1N} \times \mathcal{M}_{\{L+1\}},
    \label{eqnf3}
     \end{align}
     and the decoder 
      \begin{align*}
         g_{\{ L+1\},L+1}^{1:k_1N}: \mathcal{M}_{\{L+1\}}\times \mathcal{B}_{L+1}^{k_1N} \to \mathcal{U}^{k_1N}.
     \end{align*} 
From the case $|\mathbb{A}|=1$, there exists $n_2$ such that for $n\geq n_2$, the message rate 
satisfies  \begin{align}
     R_{\{L+1\}}= I(U;X|B_{L+1})+\frac{\delta}{2},\label{eqn13}
     \end{align}
     and the probability of error satisfies 
      \begin{align}
        \mathbb{P} [  g_{\{L+1 \}, L+1}^{1:k_1N}(M_{\{L+1\}}, {B}_{L+1,1:k_1}^{1:N} ) \neq {\widetilde U}^{1:N}_{1:k_1}]\leq \frac{\epsilon}{k_2},\label{eqn14}
     \end{align}
      where $(\widetilde{U}_{1:k_1}^{1:N},M_{\{L+1\}}) \triangleq f_{\{L+1\}}^{1:k_1N} ({X}_{1:k_1}^{1:N}) $.    \textcolor{black}{We now describe the encoder and decoders for the case $|\mathbb{A}|= L+1$}. 
     Let $k\triangleq ~k_1k_2$, $n_0\triangleq ~\max(n_1, n_2)$. 
     The encoder is defined in Algorithm $\ref{algm4}$, the first $L$ decoders and Decoder~$L+1$ are defined  in Algorithms $\ref{algm5}$ and $\ref{algm6}$, respectively.
     \begin{algorithm}[h]
  \caption{Encoding for the case $|\mathbb{A}|=L+1$ }
  \label{alg:encoding_2}
  \begin{algorithmic}   [1] 
    \REQUIRE Source observations $X^{1:k_1N}_{1:k_2} \triangleq (X_i^{1:k_1N})_{i \in [k_2]}$, where  $X_i^{1:k_1N}$ is of length $k_1N$ and corresponds to Block~$i\in [ k_2]$; a vector $R_{1}$ of $|\mathcal{V}_{U|X}|$ uniformly distributed bits.
    \STATE Transmit $R_{1}$ to Bob over the public channel   
 \STATE Construct $\widetilde{U}_{1:k_2}^{1:k_1N}$ using Algorithm $\ref{algm1}$ from the random variables $X^{1:k_1N}_{1:k_2}$ and $R_{1}$ 
 \STATE Define $(\widetilde{U}_{i}^{1:k_1N},M_{[L],i}) \triangleq f^{1:k_1N}_{[L]}(X_i^{1:k_1N})$ for $i \in [k_2]$
  \STATE Define $(\widetilde{U}_{i}^{1:k_1N}\!,M_{\{L+1\},i}) \! \triangleq \! f^{1:k_1N}_{\{L+1\}}( X_i^{1:k_1N})$ for $i \in [k_2]$
 \STATE Define        \begin{align}\nonumber
        & M_{[L+1]}\\\triangleq 
         & \left[M_{[L],1},  {( M_{[L],i+1}\oplus M_{\{L+1\},i})_{i=1}^{k_2-1}}, M_{\{L+1\},k_2}\right],
        \label{eqnmold}
  \end{align}
   where       $\oplus$ denotes bitwise modulo two addition, and if the two sequences in the sum have different lengths, then the shorter sequence is padded with zeros. 
        \STATE Transmit $M_{[L+1]}$ to Bob over the public channel
      \STATE Define $ f^{1:kN}_{[L+1]}( X^{1:k_1N}_{1:k_2})\triangleq (\widetilde{U}_{1:k_2}^{1:k_1N} ,M_{[L+1]})$
           \end{algorithmic}\label{algm4}
\end{algorithm} 
      \begin{algorithm}
  \caption{Decoding at Decoder $l\in[L]$}
  \label{alg:dencoding_1}
  \begin{algorithmic}   [1] 
    \REQUIRE Source observations $B^{1:k_1N}_{l,1:k_2} \triangleq (B_{l,i}^{1:k_1N})_{i \in [k_2]}$, where  $B_{l,i}^{1:k_1N}$ is of length $k_1N$ and corresponds to Block~$i\in [ k_2]$; 
    the message $M_{[L+1]}$ and $R_{1}$. 
    \STATE Create an estimate of $\widetilde U_{1}^{1:k_1N}$
\begin{align}
        &\hat{U}_{l, 1}^{1:k_1N}\triangleq g_{[L], l}^{1:k_1N}( M_{[L],1}, B_{l,1}^{1:k_1N})\label{eqnd1},
     \end{align} 
     where $M_{[L],1}$ is contained in $M_{[L+1]}$ 
 \FOR{Block $i=1$ to $k_2-1$} 
\STATE Define $(\widetilde{U}_{i}^{1:k_1N}, \hat M_{\{L+1\},i}) \triangleq f_{\{L+1\}}^{1:k_1N}(\hat U_{l,i}^{1:k_1N})$
   \STATE Since  $M_{[L+1]}$ contains $M_{[L],i+1}\oplus M_{\{L+1\},i}$, the decoder can form 
   \begin{align}\nonumber
      & \hat{U}_{l, i+1}^{1:k_1N}\\& \triangleq g_{[L], l}^{1:k_1N}\left( M_{[L],i+1}\oplus M_{\{L+1\},i} \oplus \hat M_{\{L+1\},i}, B_{l,i+1}^{1:k_1N}\right)\label{eqnd2}
   \end{align} 
     \ENDFOR
     \STATE Define $\hat{U}_{l,1:k_2}^{1:k_1N} \triangleq (\hat{U}_{l,i}^{1:k_1N})_{i\in[k_2]}$
     \STATE Define $g_{[L+1],l}^{1:kN}(M_{[L+1]}, B^{1:k_1N}_{l,1:k_2}) \triangleq \hat{U}_{l,1:k_2}^{1:k_1N} $
               \end{algorithmic}\label{algm5}
\end{algorithm}
\begin{algorithm}
  \caption{Decoding  at Decoder $L+1$
}
  \label{alg:dencoding_3}
  \begin{algorithmic} [1] 
    \REQUIRE Source observations $B^{1:k_1N}_{L+1,1:k_2} \triangleq (B_{L+1,i}^{1:k_1N})_{i \in [k_2]}$, where  $B_{L+1,i}^{1:k_1N}$ is of length $k_1N$ and corresponds to Block~$i\in [ k_2]$; 
    the message $M_{[L+1]}$ and $R_{1}$.  
\STATE    Create an estimate 
    \begin{align}
        &\hat{U}_{L+1,k_2}^{1:k_1N}\triangleq g_{\{L+1\}, L+1}^{1:kN}(  M_{\{L+1\},k_2}, B_{L+1,k_2}^{{1:k_1N}}),
        \label{eqnd3}
     \end{align}
       where $ M_{\{L+1\},k_2}$ is contained in $M_{[L+1]}$
 \FOR{Block $i=k_2-1$ to $1$} 
\STATE Define $(\widetilde{U}_{i+1}^{1:k_1N}, \hat M_{[L],i+1}) \triangleq f_{[L]}^{1:k_1N}(\hat U_{L+1,i+1}^{1:k_1N})$
   \STATE Since  $M_{[L+1]}$ contains $ M_{\{L+1\},i} \oplus M_{[L],i+1}$, the decoder can form
    \begin{align}\nonumber
      & {\hat{U}_{L+1,i}}^{1:k_1N}\triangleq g_{\{L+1\},L+1, i}^{1:k_1N}\\ & \left( M_{\{L+1\},i} \oplus M_{[L],i+1} \oplus \hat M_{[L],i+1}, B^{1:k_1N}_{L+1,i}\right). \label{eqnd4}
     \end{align} 
    \ENDFOR
    \STATE Define $\hat{U}_{L+1,1:k_2}^{1:k_1N} \triangleq (\hat{U}_{L+1,i}^{1:k_1N})_{i\in[k_2]}$
     \STATE Define $g_{[L+1],L+1}^{1:kN}(M_{[L+1]}, B^{1:k_1N}_{L+1,1:k_2}) \triangleq \hat{U}_{L+1,1:k_2}^{1:k_1N} $                    \end{algorithmic}\label{algm6}
\end{algorithm}

Since $\widetilde{U}_{1:k_2}^{1:k_1N}$ is constructed with Algorithm $\ref{algm1}$ as indicated in Line $2$ of Algorithm $\ref{algm4}$, by $(\ref{eqn5})$, Equation $(\ref{eqn4})$ holds. 

Next, the overall communication rate is the rate of $M_{[L+1]}$ plus the rate of $R_{1}$,
\begin{align*}
     R_{[L+1]} 
& = \frac{\log|\mathcal M_{[L+1]}|+|\mathcal{V}_{U|X}|}{kN}
   \\
   &\stackrel{(a)}\leq \frac{ k_1 k_2 N \left( \displaystyle \max_{ l\in [L+1]} I(U;X|B_l)+\frac{\delta}{2} \right)  }{kN}\\ & 
  \phantom{--} +  \frac{ k_1 N \left(I(U;X|B_{L+1})+\frac{\delta}{2} \right)}{kN}+\frac{|\mathcal{V}_{U|X}|}{kN}
   \\
      &\leq  \left(1+ \frac{1}{k_2}\right)  \left( \max_{ l\in [L+1]} I(U;X|B_l)+\frac{\delta}{2} \right)  + \frac{|\mathcal{V}_{U|X}|}{kN}
   \\
     & \xrightarrow{N \to +\infty}\left(1+ \frac{1}{k_2}\right)  \left( \max_{ l\in [L+1]} I(U;X|B_l)+\frac{\delta}{2} \right) \\&\phantom{----}+\frac{H(U|X)}{k_2}\\
      & \stackrel{(b)} \leq  \max_{ l\in [L+1]} I(U;X|B_l)+\delta,
     \end{align*}  
 where $(a)$ holds by the definition of $M_{[L+1]}$ in \eqref{eqnmold}, $(\ref{eqn11})$, and $(\ref{eqn13})$,  the limit holds by \cite[Lemma~7]{chou2016polar}, and $(b)$ holds by definition of $k_2$.

   We now analyze the probability of error. The  probability of error at Decoder $l \in [L]$ is
       \begin{align}\nonumber
       & \mathbb{P} [{\widetilde U}^{1:k_1N}_{1:k_2}\neq  \hat{U}_{l,1:k_2}^{1:k_1N} ] \\\nonumber
    &= \displaystyle\sum_{i=1}^{k_2}\mathbb{P} [{\widetilde U}_{i}^{1:k_1N}\neq  \hat{U}_{l,i}^{1:k_1N} ,  {\widetilde U}_{1:i-1}^{1:k_1N}=  \hat{U}_{l,1:i-1}^{1:k_1N}  ]\\\nonumber
    &\leq  \displaystyle\sum_{i=1}^{k_2}\mathbb{P} [{\widetilde U}_{i}^{1:k_1N}\neq  \hat{U}_{l,i}^{1:k_1N} |  {\widetilde U}_{1:i-1}^{1:k_1N}=  \hat{U}_{l,1:i-1}^{1:k_1N}  ]\\
    &       \leq \epsilon,\label{eqne1}
     \end{align}
     where the last inequality holds by $(\ref{eqn12})$.

 Similarly, the  probability of error at Decoder $L+1$ is\begin{align}\nonumber
       & \mathbb{P} [{\widetilde U}^{1:k_1N}_{1:k_2}\neq  \hat{U}_{L+1,1:k_2}^{1:k_1N} ] \\\nonumber
         &= \displaystyle\sum_{i=0}^{k_2-1}\mathbb{P} [{\widetilde U}_{k_2-i}^{1:k_1N}\! \neq \! \hat{U}_{L+1,k_2-i}^{1:k_1N} ,  {\widetilde U}_{k_2-i+1:k_2}^{1:k_1N}\! = \!  \hat{U}_{L+1,k_2-i
         +1:k_2}^{1:k_1N}  ]\\\nonumber
              &\leq  \displaystyle\sum_{i=0}^{k_2-1}\mathbb{P} [{\widetilde U}_{k_2-i}^{1:k_1N}\! \neq \!  \hat{U}_{L+1,k_2-i}^{1:k_1N} |  {\widetilde U}_{k_2-i+1:k_2}^{1:k_1N}\! = \!  \hat{U}_{L+1,k_2-i
         +1:k_2}^{1:k_1N}  ]\\ 
                       &      \leq \epsilon,\label{eqne2}
               \end{align}
               where the last inequality holds by   $(\ref{eqn14})$.
\section{Proof of Proposition \ref{thm2}}\label{sec:proof of theorem 1}

\subsection{Coding Scheme}\label{seccodingscheme}

Using  Theorem $\ref{thm1}$, consider a coding scheme $\mathcal{C}$ as defined in Definition $\ref{dfn1}$ that achieves the public communication rate   
\begin{align}
    R_p=\displaystyle\max_{ \mathcal{A}\in \mathbb{A}}I(U;X|Y_{\mathcal{A}}) + \delta
    \label{eqnrate1}
\end{align} for  the joint distribution $(p_{U^{1:N}X^{1:N}Y_{\mathcal{A}}^{1:N}})_{\mathcal{A}\in \mathbb{A}} \triangleq  (\prod_{i=1}^N p_{U|X} p_{XY_{\mathcal{A}}})_{\mathcal{A}\in \mathbb{A}}$, where $(p_{XY_{\mathcal{A}}})_{\mathcal{A}\in \mathbb{A}}$ corresponds to the source model.
For $\mathcal{A}\in \mathbb{A}$, let $\widetilde p_{U^{1:N}X^{1:N}Y_{\mathcal{A}}^{1:N}M}$ be the probability distribution of the random variables $(\widetilde U^{1:N},  X^{1:N},Y_{\mathcal{A}}^{1:N},M)$ induced by the coding scheme $\mathcal{C}$.

Then, repeat $t$ times and independently the coding scheme~$\mathcal{C}$, and denote the random variables induced by these $t$ repetitions by $(\widetilde U^{1:tN},  X^{1:tN},Y_{\mathcal{A}}^{1:tN},M^t)$. 
Additionally, the estimate of $\widetilde U^{1:tN}$ at Decoder~$\mathcal{A}\in \mathbb{A}$ is denoted by $\hat U^{1:tN}_{\mathcal{A}}$.

Consider the function  
$F : \mathcal R \times \{0, 1\}^{tN} \longrightarrow \{0, 1\}^{tNR_s}, (R,X) \mapsto (R\odot X)_{tNR_s}$, 
 where $R_s$ will be defined later, $\mathcal{R}\triangleq\{0,1\}^{tN} \backslash \{\mathbf{0}\}$,  $\odot$ is the multiplication in the finite field $\textup{GF}(2^{tN})$, and $(\cdot)_{d}$ selects the $d$ most significant bits. Note that by \cite{bellare2012semantic} $\mathcal{F} \triangleq \{X\mapsto F(R,X) : R \in \mathcal{R}\}$ forms a family of two-universal hash functions.

The encoder forms $S\triangleq F(R, \widetilde U^{1:tN})$, where $R$ represents the uniformly random choice of the hash function in $\mathcal{F}$  and is transmitted over the public communication channel. Hence, $R$ is made available to all parties.\footnote{Note that by a hybrid argument, e.g., \cite{bellare2012semantic}, the communication rate associated with the transmission of $R$ can be made negligible.} Then, Decoder $\mathcal{A}\in \mathbb{A}$ forms  $\hat S({\mathcal{A}}) \triangleq F(R, \hat U_{\mathcal{A}}^{1:tN})$ for $\mathcal{A}\in \mathbb{A}$.

\subsection{Coding Scheme Analysis}\label{analysis}
For the coding scheme analysis, we will need the following two lemmas.
\begin{lem}[{{\cite {renner2008security}}}]
    Let  $U$
    and $Z$ be random variables distributed according to $p_{UZ}$  over $\mathcal{U} \times \mathcal Z$. Let $F : \mathcal R \times \{0, 1\}^k \longrightarrow \{0, 1\}^r$ be a two universal hash function.
Let $R$ be uniformly distributed over $\mathcal R$.
Then, we have for any $z \in Z$,  $\epsilon >0$,
\begin{align}
 \mathbb V(p_{F{(R, U)},R,Z}, p_{U_{\mathcal{K}}} p_{U_{\mathcal{F}}} p_Z) \leq 2\epsilon +\sqrt{2^{
r-{H_{\infty}^{\epsilon}(p_{UZ} |p_Z )}}}, \label{eqn24}
\end{align}
where $p_{U_{\mathcal K}}$ and $p_{U_F }$ are the uniform distributions over $\{0, 1\}^r$ and $\mathcal R$, respectively, and the $\epsilon$-smooth conditional min-entropy is defined as  \cite {renner2008security},
\begin{align*}
    H_{\infty}^{\epsilon}(p_{U Z}|p_{Z}) \triangleq \!\! \displaystyle \max_{{r_{U Z} \in \mathcal{B}^{\epsilon}(p_{UZ}) }}\displaystyle \min_{{ z \in \textup{supp}(p_{Z}) }}\displaystyle \min_{{ u \in\mathcal U} }\log\left[\frac{p_{Z}(z)}{r_{UZ}(u,z)}\right]\!.
    \end{align*}
    where $\textup{supp}(p_{Z})\triangleq \{z\in \mathcal{Z}:p_{Z}(z)>0\}$ and $\mathcal{B}^{\epsilon}(p_{UZ})\triangleq \{(r_{UZ}:\mathcal U\times \mathcal{Z}\rightarrow \mathbb{R}_+):\mathbb V(p_{UZ}, r_{UZ})\leq \epsilon\}$.
\label{lem2}
\end{lem}
  
  \begin{lem}
    [{\cite{holenstein2011randomness}}] 
    Let $p_{X^L Z^L}\triangleq \displaystyle \prod_{i=1}^L p_{X^i Z^i} $  be a probability distribution over $\mathcal{X}^L \times \mathcal{Z}^L $. For any~$\delta>0$,\\ $$H_\infty^\epsilon(p_{X^L Z^L}|p_{Z^L})\geq H(X^L|Z^L)-L\delta,$$ where $\epsilon =2^{-\frac{L\delta^2}{2\log^2(|\mathcal{X}|+3)}}$.\label{lem3}
\end{lem}          

         \subsubsection{Analysis of Security and 
 Uniformity         } \label{secanalsecu}
         Let $p_{U_{\mathcal S}}$ and $p_{U_F }$ be the uniform distributions over $\{0, 1\}^{tNR_s}$ and $\mathcal R$ respectively. 

          For  $\mathcal{U}\in \mathbb{U}$, we have
\begin{align}\nonumber
& \mathbb V(\widetilde p_{F{(R,\widetilde U^{1:tN})
},R,Y^{1:tN}_{\mathcal{U}}M^t},   p_{U_{\mathcal{S}}} p_{U_{\mathcal{F}}} \widetilde p_{{Y}^{1:tN}_{\mathcal{U}}M^t})\\\nonumber
&\stackrel{(a)}\leq 2\epsilon +\sqrt{2^{
tNR_s-{H_{\infty}^{\epsilon}({\widetilde p_{
U^{1:tN}_{}Y^{1:tN}_{\mathcal{U}}M^t}} |{\widetilde p_{Y^{1:tN}_{\mathcal{U}}M^t})} }}}\\
&\stackrel{(b)}\leq 2\epsilon \!+\!\sqrt{2^{
t(NR_s-H({ \widetilde U^{1:N}|  Y^{1:N}_{\mathcal{U}}M})\!+\!\delta)}} \nonumber \\
&  \leq 2\epsilon \!+\!\sqrt{2^{
t(N(R_s-\min_{\mathcal{U}\in  \mathbb{U}}H(U|Y_{\mathcal{U}}) + \max_{\mathcal{A}\in  \mathbb{A}}H(U|Y^{}_{\mathcal{A}}))+o(N)+\delta)}}
\label{eqn32}
\end{align}
where $(a)$ holds  by Lemma \ref{lem2} with $\epsilon  \triangleq 2^{-\frac{t\delta^2}{2\log^2(|\mathcal{U}|+3)}}$, $\delta>0$, $(b)$  holds by  Lemma \ref{lem3}, $(c)$ holds because for any for $\mathcal{U}\in \mathbb{U} $, we have
    \begin{align}\nonumber
        & H({ \widetilde U^{1:N}|  Y^{1:N}_{\mathcal{U}}M}) \\\nonumber &= H( \widetilde U^{1:N} |  Y^{1:N}_{\mathcal{U}})- I(\widetilde U^{1:N} ;M| Y^{1:N}_{\mathcal{U}}) \\\nonumber
        &\geq H( \widetilde U^{1:N} |  Y^{1:N}_{\mathcal{U}})- H(M)\\\nonumber
   &= H( \widetilde U^{1:N} |  Y^{1:N}_{\mathcal{U}})- N\max_{\mathcal{A}\in  \mathbb{A}}I(U^{}_{};X|Y^{}_{\mathcal{A}})-o(N)
    \\  \nonumber    &= N(H(U|Y_{\mathcal{U}}) -  H(U|X)- \max_{\mathcal{A}\in  \mathbb{A}}I(U^{}_{};X|Y^{}_{\mathcal{A}}))-o(N) \\  \nonumber    &= N(H(U|Y_{\mathcal{U}}) - \max_{\mathcal{A}\in  \mathbb{A}}H(U|Y^{}_{\mathcal{A}}))-o(N)\\     & \geq N(\min_{\mathcal{U}\in  \mathbb{U}}H(U|Y_{\mathcal{U}}) - \max_{\mathcal{A}\in  \mathbb{A}}H(U|Y^{}_{\mathcal{A}}))-o(N)
            ,
   \label{eqn33} \end{align}where the second equality holds because  $H(M) \leq |M|$ and by \eqref{eqnrate1}, the third equality holds by \eqref{eqlemma2}, the fourth inequality holds by the Markov chain $U -X -Y_{\mathcal{A}}$.

        Next, define the  output length of the hash function as 
   \begin{align}
      tNR_s\triangleq t(N(\displaystyle{\min_{{\mathcal{U}}\in  \mathbb{U}}}H({U|{Y_{\mathcal{U}}}})\!-\! {\displaystyle \max_{\mathcal{A}\in  \mathbb{A}}}H(U|Y_{\mathcal{A}}))-o(N)-2\delta).
 \label{eqnr}  \end{align}  
    Then, for any $\mathcal{U}\in  \mathbb{U}$ and $i \in [k]$, \eqref{eqn32}, \eqref{eqn33}, and \eqref{eqnr} yield
    \begin{align}
     \mathbb V(\widetilde p_{F{(R,\widetilde U^{1:tN})
},R,Y^{1:tN}_{\mathcal{U}}M^t},   p_{U_{\mathcal{S}}} p_{U_{\mathcal{F}}} \widetilde p_{{Y}^{1:tN}_{\mathcal{U}}M^t})\leq 2\epsilon +\sqrt{2^{
-t\delta}}.
\label{eqn36}
\end{align}  

Then, by {\cite [Lemma 2.7]{csiszar2011information}}, for $t$ large enough and $\epsilon$ small enough,  we conclude from \eqref{eqn36}   
with $f:x \mapsto x\log (2^{tN}/x)$ 
\begin{align*}
\displaystyle\max_{\mathcal{U}\in  \mathbb{U}} 
I(S; Y^{1:tN}_{\mathcal{U}}M^tR)&\leq f( 2\epsilon +\sqrt{2^{
-t\delta}}),\\
\log |\mathcal{S}|-H(S)&\leq   f(2\epsilon +\sqrt{2^{
-t\delta}}).
 \end{align*} 
 \subsubsection{Analysis of Reliability}  We have
   \begin{align}\nonumber
          \max_{\mathcal{A} \in \mathbb{A}}  \mathbb{P}[S\neq \hat S(\mathcal{A})] \nonumber  &= \max_{\mathcal{A} \in \mathbb{A}} \mathbb{P}[F(R,\widetilde U^{1:tN}_{})\neq F(R,\hat U^{1:tN}_{\mathcal{A}}) ]\\\nonumber
       &\leq \max_{\mathcal{A} \in \mathbb{A}} \mathbb{P}[\widetilde U^{1:tN}_{}\neq \hat U^{1:tN}_{\mathcal{A}} ]
         \\  &\leq t \max_{\mathcal{A} \in \mathbb{A}} \mathbb{P}[\widetilde U^{1:N}_{}\neq \hat U^{1:N}_{\mathcal{A}} ],\label{eqn39}
         \end{align}where the last inequality holds by the union bound.            
         \section{Proof of Theorem 
         \ref{thm3}}\label{proof of theorem 2}  We first propose a simple extension of Theorem~$\ref{thm1}$ in Section \ref{part I}.  Then, in Section \ref{part II}, we describe our coding~scheme for the achievability of Theorem \ref{thm3}.
       
      \subsection{Preliminary Result: Extension of Theorem \ref{thm1}}\label{part I}

Fix a \textcolor{black}{sequence} of joint probability distribution  $(p_{V_{}|UX_{}}p_{U_{}|X_{}}p_{XY_{\mathcal{A}}})_{{\mathcal{A}}\in \mathbb{A}}$. Similar to Section \ref{sec:auxiliary result},  we assume that the participants and dealer observe $k$ blocks of lengths $N$, i.e., 
      $Y^{1:N}_{{\mathcal{A}},1:k} \triangleq (Y^{1:N}_{{\mathcal{A}},i})_{i \in [k]}$ is the source observation  at Decoder ${\mathcal{A}}\in \mathbb{A}$ and  $X^{1:N}_{1:k} \triangleq (X^{1:N}_{i})_{i \in [k] }$ is the source observation at the encoder  with  $p_{X^{1:N}_{1:k}Y^{1:N}_{{\mathcal{A}},1:k}} \triangleq  \prod_{i=1}^{kN} p_{XY_{\mathcal{A}}}$, ${\mathcal{A}}\in\mathbb{A}$.
      
\begin{thm}  Let $N \triangleq 2^n$,  $n\in \mathbb{N}$. Define $\mathcal{M}_U \triangleq [2^{NR_U}]$, $ \mathcal{M}_V \triangleq [2^{NR_V}]$. For any $\epsilon, \delta 
                     >0$, there exist $k,n_0 \in \mathbb{N}$ such that for any $n\geq n_0$, 
                there exists an encoder
     \begin{align*}
         &f
         : \mathcal{X}^{kN}_{} \to \mathcal{V}^{kN}_{}\times \mathcal{M}_U\times \mathcal{M}_V,
           \end{align*}
          and $|\mathbb{A}|$ decoders
      \begin{align*}
       g_{\mathcal{A}}
       : \mathcal{M}_U\times \mathcal{M}_V \times \mathcal{Y}_{{\mathcal{A}}}^{kN} \to \mathcal{V}^{kN}, \forall {\mathcal{A}}\in \mathbb{A}, 
     \end{align*}
such that 
    the message rate $R_U+R_V$ satisfies
      \begin{align}
      (R_U,R_V) =   \left( \max_{ \mathcal{A}\in \mathbb{A}}I(U;X|Y_{\mathcal{A}}), \displaystyle\max_{ \mathcal{A}\in \mathbb{A}}I(V;X|UY_{\mathcal{A}}) \right) + \frac{\delta}{2},
        \label{eqnnewrate}
     \end{align} and the probability of error at Decoder ${\mathcal{A}}\in \mathbb{A}$ satisfies
      \begin{align}
       \mathbb{P} [\widetilde{V}^{1:N}_{1:k}\neq  g_{\mathcal{A}}
      (
      M_U,M_V, {Y}_{{\mathcal{A}},1:k}^{1:N}
      )]<\epsilon, \forall  {\mathcal{A}}\in \mathbb{A},\label{eqnnewerror}
     \end{align} 
     where     $(\widetilde{V}^{1:N}_{1:k}, M_U, M_V) \triangleq f (X^{1:N}_{1:k})$. Furthermore, the  encoding and decoding complexities are 
${O}(kN \log N)$, since the scheme operates over $k$ blocks of length $N$, 
each of which inherits the ${O}(N \log N)$ complexity of 
source polarization~\cite{arikan2010source}.

 Additionally, we have the following two properties.
\begin{itemize}
    \item The probability distribution induced by the encoding scheme $\widetilde p_{U^{1:N}_{i}V^{1:N}_{i}X^{1:N}_{i}}$,  $i\in [k]$,  satisfies  
 \begin{align}
& \mathbb{V} \left( \widetilde{p}_{U_{1:k}^{1:N}V_{1:k}^{1:N}X_{1:k}^{1:N}} , p_{U^{1:N}_{}V_{1:k}^{1:N}X^{1:N}_{1:k}} \right) \leq  \delta_N^{(1)}, \label{eqnref41}
	\end{align}
 where $\delta_N^{(1)}$ is defined in Theorem \ref{thm1}. 
 \item For any $\mathcal{U} \in \mathbb{U}$,
    \begin{align}
&        H(\widetilde V^{1:N}_{1:k} | Y^{1:N}_{\mathcal{U},1:k} \widetilde U^{1:N}_{1:k} ) \nonumber \\
        & \phantom{---}\geq Nk (H(V|Y_{\mathcal{U}}U) -  H(V|UX)) - o(Nk). \label{eqentropyth2}
    \end{align}
 \end{itemize}

             \label{thm4}
     \end{thm}
 Theorem \ref{thm4} is obtained by applying twice Theorem~\ref{thm1}. Specifically, the first time, $(\widetilde{U}^{1:N}_{1:k}, M_U) $ is obtained from $ X^{1:N}_{1:k}$ such that $R_U = \max_{ \mathcal{A}\in \mathbb{A}}I(U;X|Y_{\mathcal{A}}) + \delta /2$ and $\widetilde{U}^{1:N}_{1:k}$ can be reconstructed from $(Y^{1:N}_{{\mathcal{A}},1:k} , M_U)$, $\mathcal{A}\in \mathbb{A}$. The second time, Theorem~\ref{thm1} is applied with the substitution $Y^{1:N}_{{\mathcal{A}},1:k} \leftarrow (\widetilde{U}^{1:N}_{1:k},Y^{1:N}_{{\mathcal{A}},1:k})$ to obtain $(\widetilde{V}^{1:N}_{1:k}, M_V) $ such that $R_V = \displaystyle\max_{ \mathcal{A}\in \mathbb{A}}I(V;X|UY_{\mathcal{A}})  + \delta /2$ and $\widetilde{V}^{1:N}_{1:k}$ can be reconstructed from $(Y^{1:N}_{{\mathcal{A}},1:k} ,\widetilde{U}^{1:N}_{1:k}, M_V)$, $\mathcal{A}\in \mathbb{A}$. Details are omitted due to the similarity with the proof of Theorem~\ref{thm1}.

\subsection{Coding Scheme}\label{part II}
\subsubsection{Coding Scheme}
 Repeat $t$ times and independently the coding scheme of Theorem \ref{thm4}, the random variables induced by these $t$ repetitions are denoted by $(\widetilde U^{1:tN},\widetilde V^{1:tN},  X^{1:tN},Y_{\mathcal{A}}^{1:tN},M_U^t,M_V^t)$. 
Additionally, the estimate of $\widetilde V^{1:tN}$ at Decoder~$\mathcal{A}\in \mathbb{A}$ is denoted by $\hat V^{1:tN}_{\mathcal{A}}$.

Similar to Section \ref{seccodingscheme}, consider a hash function 
$F : \mathcal R \times \{0, 1\}^{tN} \longrightarrow \{0, 1\}^{tNR_s}$,  
 where $R_s$ will be defined later. The encoder forms 
$S\triangleq F(R, \widetilde V^{1:tN})$, where $R$ represents the uniformly random choice of the hash function in a family of two-universal hash functions. Then, Decoder $\mathcal{A}\in \mathbb{A}$ forms  $\hat S({\mathcal{A}}) \triangleq F(R, \hat V_{\mathcal{A}}^{1:tN})$ for $\mathcal{A}\in \mathbb{A}$.

         \subsubsection{Coding Scheme Analysis
         }
         The coding scheme analysis is similar to Section \ref{analysis}, we thus omit the details and only highlight the main difference by showing the counterpart of \eqref{eqn36}. 
         Let $p_{U_{\mathcal S}}$ and $p_{U_F }$ be the uniform distributions over $\{0, 1\}^{tNR_s}$ and $\mathcal R$ respectively. For  $\mathcal{U}\in \mathbb{U}$, similar to \eqref{eqn32}, we~have
\begin{align}\nonumber
& \mathbb V(\widetilde p_{F{(R,\widetilde V^{1:tN}_{})
},R,Y^{1:tN}_{\mathcal{U}}M_U^tM_V^t},   p_{U_{\mathcal{S}}} p_{U_{\mathcal{F}}} \widetilde p_{{Y}^{1:tN}_{\mathcal{U}}M_U^tM_V^t})\\
&\leq 2\epsilon+\sqrt{2^{
t(NR_s-H(
\widetilde V^{1:N}_{}| Y^{1:N}_{\mathcal{U}}M_UM_V)+\delta)}},
\label{eqn59}
\end{align}
where  we have
    \begin{align}\nonumber
& H(\widetilde V^{1:N}_{}| Y^{1:N}_{\mathcal{U}}M_UM_V) \\\nonumber 
& \geq  H(\widetilde V^{1:N}_{}|\widetilde U^{1:N}  Y^{1:N}_{\mathcal{U}}M_UM_V) \\\nonumber 
&= H(
\widetilde V^{1:N}_{}|\widetilde U^{1:N} Y^{1:N}_{\mathcal{U}}M_V) \\\nonumber
   &\geq H(
\widetilde V^{1:N}_{}|\widetilde U^{1:N} Y^{1:N}_{\mathcal{U}}) - H(M_V)
    \\\nonumber
   &= H(
\widetilde V^{1:N}_{}|\widetilde U^{1:N} Y^{1:N}_{\mathcal{U}}) - \displaystyle\max_{ \mathcal{A}\in \mathbb{A}}I(V;X|UY_{\mathcal{A}}) - o(N)\\  \nonumber    &\geq  N(H(V|UY_{\mathcal{U}}) -  H(V|UX)- \displaystyle\max_{ \mathcal{A}\in \mathbb{A}}I(V;X|UY_{\mathcal{A}})\!-\!o(N) \\  \nonumber    &= N(H(V|UY_{\mathcal{U}}) - \max_{\mathcal{A}\in  \mathbb{A}}H(V|UY^{}_{\mathcal{A}}))-o(N)\\     & \geq N(\min_{\mathcal{U}\in  \mathbb{U}}H(V|UY_{\mathcal{U}}) - \max_{\mathcal{A}\in  \mathbb{A}}H(V|UY^{}_{\mathcal{A}}))-o(N)
            ,
   \label{eqn33b} \end{align}
    where the first equality holds because $M_U$ is a function of $\widetilde U^{1:N}$, the third inequality holds by \eqref{eqentropyth2}, the third equality holds by the Markov Chain $ (U,V)-X -Y_{\mathcal{A}}$.

        Next, define the  output length of the hash function as 
   \begin{align}
     & tNR_s \nonumber\\ &  \triangleq tN(\displaystyle{\min_{{\mathcal{U}}\in  \mathbb{U}}}H({V|{UY_{\mathcal{U}}}})- {\displaystyle \max_{\mathcal{A}\in  \mathbb{A}}}H(V|UY_{\mathcal{A}}))-o(tN)-2t\delta.
 \label{eqnrr}  \end{align}  
    Then, for any $\mathcal{U}\in  \mathbb{U}$, $(\ref{eqn59})$ becomes
    \begin{align}\nonumber
    & \mathbb V(\widetilde p_{F{(\!R,\widetilde V^{1:tN}\!)
},R,U^{1:tN}Y^{1:tN}_{\mathcal{U}}M_U^tM_V^t}, p_{U_{\mathcal{S}}} p_{U_{\mathcal{F}}} \widetilde p_{U^{1:tN}Y^{1:tN}_{\mathcal{U}}M_U^tM_V^t}\!)\\
&\leq 2\epsilon +\sqrt{2^{
-t\delta}},
\label{eqn62}
\end{align} 
Finally, uniformity and security follows from \eqref{eqn62} similar to Section \ref{secanalsecu}.

\section{Concluding Remarks}\label{concluding remarks} We considered secret sharing from correlated randomness and rate-limited public communication, and proposed the first explicit coding scheme able to handle arbitrary access structures. Our construction relies on vector quantization coupled with distribution approximation to handle the reliability constraints, followed by universal hashing to handle the security constraints.
Our results also yield explicit capacity-achieving coding schemes for one-way rate-limited secret-key generation, for arbitrarily correlated random variables and without the need of a pre-shared secret. 
While our study focuses on binary alphabets, the proposed coding scheme can be extended to discrete alphabets with prime cardinality by applying \cite[Lemma 7]{chou2016polar} in place of \cite[Lemma 1]{chou2013polar}, and by using \cite[Lemma 6]{chou2016polar} instead of \cite{arikan2010source} when analyzing the cardinality of the polarized index sets. Extending our results to continuous alphabets, however, requires different techniques, e.g., polar lattices \cite{liu2018construction}. Furthermore, while we considered the asymptotic regime, the finite-block length regime remains an open problem, recent works in this direction for compound settings include \cite{rana2023short,sultana2023secret}.
         \bibliographystyle{IEEEtran}       \bibliography{key}  
\end{document}